\documentclass{article}

\PassOptionsToPackage{numbers,sort&compress}{natbib}

\usepackage[main, preprint]{neurips_2026}

\usepackage[utf8]{inputenc}
\usepackage[T1]{fontenc}
\usepackage{url}
\usepackage{booktabs}
\usepackage{amsfonts}
\usepackage{nicefrac}
\usepackage{microtype}
\usepackage{xcolor}
\usepackage{graphicx}
\usepackage{tabularx}
\usepackage{subcaption}
\usepackage{makecell}
\usepackage{multirow}
\usepackage{colortbl}
\usepackage{amsmath}
\usepackage{amssymb}
\usepackage{mathtools}
\usepackage{amsthm}
\usepackage[hidelinks]{hyperref}
\usepackage{enumitem}
\usepackage[capitalize,noabbrev]{cleveref}
\usepackage{caption}
\usepackage{wrapfig}
\theoremstyle{plain}
\newtheorem{theorem}{Theorem}[section]
\newtheorem{proposition}[theorem]{Proposition}

\theoremstyle{definition}

\theoremstyle{remark}

\title{Rethinking Entropy Allocation in LLM-based ASR: Understanding the Dynamics between Speech Encoders and Large Language Models}

\author{%
  Yuan Xie{\equalcontribfirst} , \enspace
  Jiaqi Song{\equalcontrib} , \enspace
  Guang Qiu, \enspace
  Xianliang Wang, \enspace
  Ming Lei, \enspace
  Jie Gao, \enspace
  Jie Wu\correspondingauthor \\[4pt]
  LLM Team, Advanced Intelligence Systems Group, NIO \\[3pt]
  \texttt{\{ryan.xie2, jiaqi.song2\}@nio.com}
}

\begin{document}

\maketitle

\begin{abstract}
Integrating large language models (LLMs) into automatic speech recognition (ASR) has become a dominant paradigm. Although recent LLM-based ASR models have shown promising performance on public benchmarks, it remains challenging to balance recognition quality with latency and overhead, while hallucinations further limit real-world deployment. In this study, we revisit LLM-based ASR from an entropy allocation perspective and introduce three encoder-side diagnostics to characterize how training paradigms shape uncertainty reduction across the speech encoder--LLM interface. To remedy entropy-allocation inefficiencies in prevailing approaches, we propose a capability-boundary-aware multi-stage training strategy that targets parameter efficiency and robustness to hallucinations. Specifically, we redesign the pretraining strategy to alleviate the speech-text modality gap, and further introduce an iterative asynchronous SFT stage between alignment and joint SFT to preserve functional decoupling and constrain encoder representation drift. Experiments on various benchmarks show that our method achieves competitive performance against state-of-the-art models using only 2.3B parameters, while also effectively mitigating hallucinations through our decoupling-oriented design.
\end{abstract}

\section{Introduction}
\label{sec1}

With the rapid advancement of large language models (LLMs), the mainstream automatic speech recognition (ASR) paradigm is shifting from traditional architectures~\citep{graves2006connectionist,graves2012sequence,chorowski2015attention,chan2016listen} toward LLM-based frameworks. Recent LLM-based ASR (hereafter LLM-ASR) models, including Seed-ASR~\citep{bai2024seed}, FireRedASR-LLM~\citep{xu2025fireredasr}, Voxtral Mini Transcribe~\citep{liu2025voxtral}, Fun-ASR~\citep{an2025funaudio}, and Qwen3-ASR~\citep{shi2026qwen3}, are specialized for transcription. By leveraging LLMs' world knowledge and contextual reasoning to resolve semantic ambiguity, LLM-ASR models have achieved promising results on public benchmarks. In contrast to Large Audio-Language Models (LALMs)~\citep{chu2024qwen2,wu2025step,goel2025audio}, which target broad audio understanding tasks beyond ASR, LLM-ASR is more suitable for industrial speech applications, where accuracy, latency, computational overhead, and controllability are all critical concerns, and thus serves as the focus of this study.

Despite strong benchmark performance, LLM-ASR still faces two practical challenges in real-world deployment. 
First, efficiency and recognition quality remain difficult to balance. In lightweight settings, LLM-ASR models suffer not only from expected scaling-down degradation, but also from the inherent cost of bridging the speech-text modality gap, which consumes non-trivial model capacity~\citep{aghajanyan2023scaling,zhang2026instruction} and imposes a disproportionate tax on smaller models~\citep{endo2025downscaling}. 
The second challenge is hallucination. During joint training, the encoder is susceptible to being dominated by the LLM's gradients, which can induce representation drift: encoder representations may increasingly encode linguistic shortcuts at the expense of acoustic fidelity, thereby amplifying hallucination risk~\citep{bai2024hallucination,zhou2024mitigating}.

\begin{figure}[!htbp]
    \centering
    \includegraphics[width=0.92\linewidth]{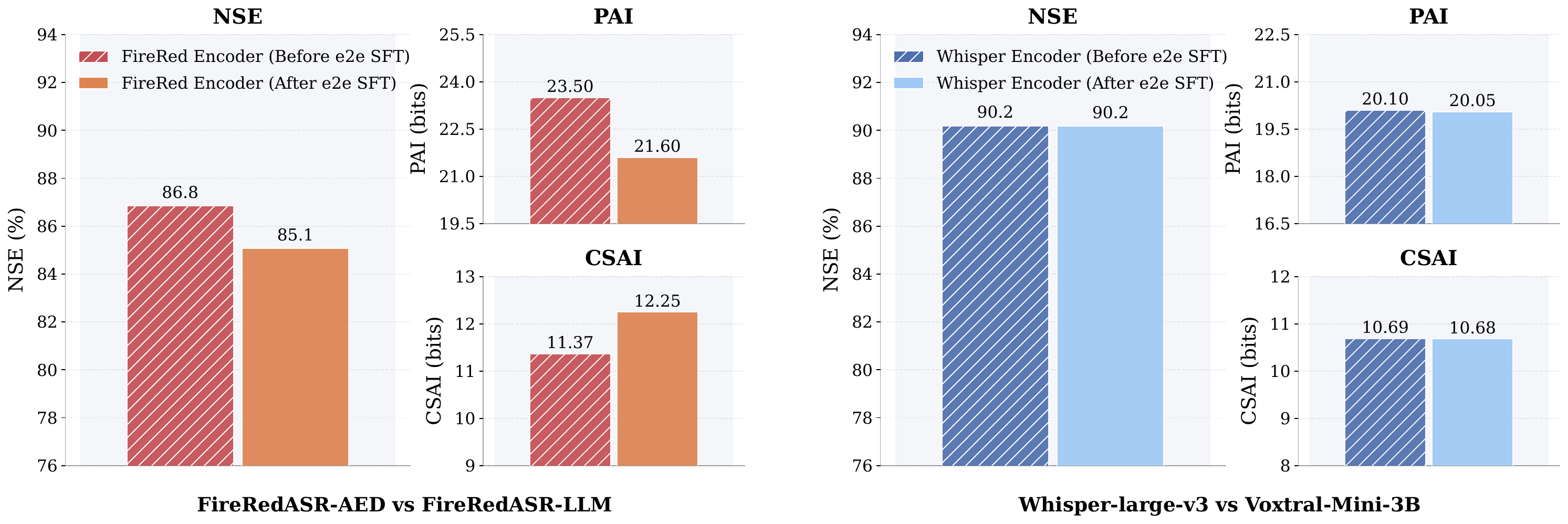}
    \caption{The shift in encoder-side metrics before and after joint training. We hold the encoder architecture fixed within each group to examine the impact of joint training on encoder representations.}
    \label{fig1}
    \vskip -5mm
\end{figure}

To step beyond empirical observation and provide a principled diagnosis of these limitations, we revisit LLM-ASR from an entropy allocation perspective, viewing ASR as compressing high-entropy speech signals into low-entropy linguistic symbols. From this perspective, we propose a set of encoder-side metrics (defined in Section~\ref{sec3.2}) to diagnose how different training paradigms shape the encoder–LLM interface. %This analysis offers a principled way to diagnose prevailing training paradigms, exposing their inefficiencies in parameter utilization efficiency and hallucination suppression. 
Specifically, normalized spectral entropy (NSE) quantifies the spectral entropy of encoder representations, while phonetic accessible information (PAI) and conditional semantic accessible information (CSAI) serve as probe-inspired proxies for linearly accessible phonetic information and residual semantic information beyond the phonetic proxy, respectively. 
Figure~\ref{fig1} compares the metric shifts before and after end-to-end joint training across two representative ASR model families, revealing distinct behavioral regimes. 
From FireRedASR-AED to FireRedASR-LLM with the same inherited Conformer encoder, joint training markedly reduces NSE, indicating increased encoder-side compression. Yet, the concurrent decrease in PAI and increase in CSAI suggest that this compression is associated with a shift away from phonetic specialization toward semantic accessibility, consistent with the representation drift discussed above. In contrast, Voxtral’s Whisper encoder shows minimal changes during joint training, with all metrics remaining largely stable. Compared to FireRedASR, it exhibits higher NSE but lower PAI and CSAI, indicating a lighter encoder-side entropy-reduction load and a greater share of residual uncertainty left to the LLM. While this mitigates representation drift, it increases downstream capacity demand and forces the LLM to resolve phonetic uncertainties outside its primary modeling strengths, leading to reduced parameter efficiency.

These observations reveal two suboptimal regimes in prevailing LLM-ASR systems: one exhibits representation drift and over-reliance on semantic priors, while the other offloads residual uncertainty to the LLM at the cost of parameter efficiency. 
Motivated by this diagnosis, we propose a capability-boundary-aware design principle that refines multi-stage training. We redesign the pretraining strategy to ground encoder representations in phonemes at a structural granularity approaching text tokens, which reduces modality alignment cost and separates acoustic modeling from premature semantic anchoring. We further introduce an iterative asynchronous SFT (IA-SFT) stage between alignment and joint SFT, which deepens the speech-text modality alignment while explicitly preserving functional decoupling to avoid representation drift. 
This refined multi-stage design encourages each module to devote more of its capacity to its designated role, improving parameter efficiency and mitigating hallucinations. Our main contributions are summarized as follows.
\begin{itemize}[leftmargin=1.6em, itemsep=2pt, parsep=0pt, topsep=1pt]
\item We present an entropy-based perspective for LLM-ASR, with metrics on encoder representations that characterize how training paradigms allocate entropy reduction between encoder and LLM.
\item We propose a capability-boundary-aware multi-stage training paradigm that redesigns encoder pretraining and introduces an intermediate IA-SFT stage, jointly promoting clear module specialization and more stable joint training.
\item Our model achieves competitive performance across Mandarin, dialect, English, and code-switching benchmarks with only 2.3B parameters, while effectively mitigating hallucinations, facilitating efficient and robust real-world deployment.
\end{itemize}

\section{Related Work}
\label{sec2}

Recent advances in LLMs have redefined ASR, shifting the paradigm from acoustic transcription toward semantically informed generation. By coupling speech encoders with LLMs, LLM-ASR leverages rich world knowledge and long-context semantic modeling to enhance recognition. In recent research, the encoder–adaptor–LLM architecture has emerged as a standard backbone, with variations primarily arising from training strategies and functional priorities.

\noindent\textbf{LLM-based ASR models.}
Recent LLM-ASR models have emerged along diverse design principles. One line emphasizes lightweight adaptation: for instance, FireRedASR-LLM~\citep{xu2025fireredasr} achieves strong performance on Mandarin and English benchmarks via LoRA fine-tuning. Another line of works such as Seed-ASR~\citep{bai2024seed} and Fun-ASR~\citep{an2025funaudio} scale up training with industrial-scale speech corpora, and further leverage context-SFT and reinforcement learning (RL) to extend model capabilities. Moreover, a related trend leverages pretrained multi-modal LLM backbones to build specialized ASR models---for example, Qwen3-ASR~\citep{shi2026qwen3} built upon Qwen3-Omni~\citep{xu2025qwen3} and GPT-4o Transcribe~\citep{openai2025gpt4otranscribe} derived from GPT-4o~\citep{hurst2024gpt}. These models inherit extensive speech understanding knowledge and apply targeted optimizations for ASR capabilities. 
Additionally, general-purpose LALMs such as GLM-4-Voice~\citep{zeng2024glm}, Step-Audio 2~\citep{wu2025step}, and Kimi-Audio~\citep{ding2025kimi} can also serve as relevant baselines in ASR benchmarking, though seldom deployed for practical ASR applications. 
Although LLM-ASR models can achieve strong benchmark results, latency and deployment cost remain major constraints for real-time speech interaction. As a result, many works~\citep{an2025funaudio,shi2026qwen3} have released lightweight variants; however, these typically rely on straightforward capacity reduction without principled changes to the training paradigm, leaving a substantial performance gap relative to their larger counterparts. Moreover, the inherent modality gap between acoustic and textual representations inevitably amplifies hallucinations~\citep{peng2024survey}, posing challenges to the reliability required for real-world deployment.

\noindent\textbf{Divergent paths in uncertainty resolution.}
% Most LLM-ASR models adopt the encoder--adaptor--LLM architecture, where divergent paths in uncertainty resolution arise from differences in encoder pretraining, cross-modal alignment, and joint-training strategies. 
% During encoder pretraining, many studies employ supervised ASR objectives such as Connectionist Temporal Classification (CTC), Attention-based Encoder-Decoder (AED), or their hybrid variants~\citep{radford2022whisper,xu2025fireredasr,xia2026uni}, which direct the encoder to capture transcription-relevant acoustic structure and resolve a substantial portion of acoustic uncertainty through representations optimized for discriminability.
% By contrast, several works~\citep{bai2024seed,an2025funaudio} leverage self-supervised learning (SSL) objectives such as Best-RQ~\citep{chiu2022self} to initialize encoders, learning general-purpose acoustic representations without collapsing toward linguistic units. Such SSL pretraining is often followed by supervised objectives (e.g., AED or CTC) before LLM integration. These strategies lead to distinct internal operating regimes, analyzed in Appendix~\ref{sec:appendix_dynamics}.
Most LLM-ASR models adopt the encoder--adaptor--LLM architecture, but differ in how uncertainty is resolved across the encoder and the LLM. During encoder pretraining, supervised ASR objectives such as CTC, AED, or their hybrid variants~\citep{radford2022whisper,xu2025fireredasr,xia2026uni} encourage the encoder to learn transcription-discriminative representations and resolve much of the acoustic uncertainty before LLM integration. In contrast, SSL objectives such as Best-RQ~\citep{chiu2022self} are used in several works~\citep{bai2024seed,an2025funaudio} to learn more general-purpose acoustic representations. These choices lead to distinct internal operating regimes, analyzed in Appendix~\ref{sec:appendix_dynamics}.
Following encoder pretraining, LLM-ASR models typically proceed through an alignment stage and then joint training with the LLM under natural-language supervision~\citep{hono2024integrating}, which further shapes how uncertainty is resolved. 
% Textual supervision drives the encoder toward transcription-relevant acoustic cues, progressively suppressing acoustic variations that are weakly associated with lexical discrimination. 
As the system maps high-dimensional speech representations into low-dimensional symbolic sequences, the granularity retained in the encoder representations directly determines how much entropy reduction must be handled by the downstream LLM, thereby dictating its effective capacity. In principle, the optimal scale of the LLM-ASR model should be jointly informed by its training paradigm, yet existing approaches remain largely data-driven and pay limited attention to this relationship.

% \noindent\textbf{Divergent paths in uncertainty resolution.}
% Most LLM-ASR models adopt the encoder--adaptor--LLM architecture, but differ in how uncertainty is resolved across the encoder and the LLM. During encoder pretraining, supervised ASR objectives such as CTC, AED, or their hybrid variants~\citep{radford2022whisper,xu2025fireredasr,xia2026uni} encourage the encoder to learn transcription-discriminative representations and resolve much of the acoustic uncertainty before LLM integration. In contrast, SSL objectives such as Best-RQ~\citep{chiu2022self} are used in several works~\citep{bai2024seed,an2025funaudio} to learn more general-purpose acoustic representations. These choices lead to distinct internal operating regimes, analyzed in Appendix~\ref{sec:appendix_dynamics}.

% Following encoder pretraining, subsequent cross-modal alignment and joint training further reshape the encoder--LLM division of labor. Textual supervision improves transcription modeling but may also suppress acoustic details weakly associated with lexical discrimination. Therefore, the granularity retained by the encoder determines how much uncertainty must be handled by the downstream LLM. Existing approaches, however, remain largely empirical and pay limited attention to this relationship.

\section{Methodology}
\label{sec3}

\subsection{Preliminary}
\label{sec3.1}

\noindent\textbf{LLM-based ASR as conditional generation.}
Let $x$ denote the input speech features (e.g., log-Mel spectrograms), $t=(t_1,\ldots,t_M)$ denote the text instruction prompt, and $y=(y_1,\ldots,y_N)$ denote the output transcriptions. We formulate LLM-based ASR as conditional language modeling with parameters $\theta=\{\phi,\psi,\omega\}$, corresponding to the speech encoder, adaptor, and LLM decoder, respectively. The speech encoder $\mathcal{E}_{\phi}$ maps acoustic features to continuous representations, and the adaptor $\mathcal{A}_{\psi}$ projects them into the LLM embedding space: 
\begin{equation}
E = \mathcal{E}_{\phi}(x) \in \mathbb{R}^{L \times d_e}, \qquad 
Z = \mathcal{A}_{\psi}(\text{Down}_{r}(E)) \in \mathbb{R}^{L' \times d_m},
\label{eq:encoder_adaptor}
\end{equation}
where $L$ is the original frame length, $L'$ is the length after downsampling, $d_e$ is the encoder hidden size, $d_m$ is the LLM embedding dimension, and $\text{Down}_{r}(\cdot)$ denotes downsampling with factor $r$. 
Let $\mathcal{W}(\cdot)$ denote the LLM text token embedding operator. We then construct a unified continuous input sequence by concatenating the text prompt embeddings and the projected speech embeddings: 
\begin{equation}
\scalebox{0.95}{$
\mathbf{Z}_{\mathrm{prompt}}(x,t)
=
\big[
\mathcal{W}(t_1),\ldots,\mathcal{W}(t_M)\;;
Z_1,\ldots,Z_{L'}
\big]
\in
\mathbb{R}^{(M+L') \times d_m}.
$}
\label{eq:embedding_concat}
\end{equation}
where $\{Z_\ell\}_{\ell=1}^{L'}$ are speech embeddings produced by the encoder--adaptor stack. The LLM decoder $\mathcal{D}_\omega$ then generates the transcription autoregressively: 
\begin{equation}
p_\theta(y \mid x, t)
=
\prod_{n=1}^{N}
p_\omega\!\left(y_n \mid y_{<n}, \mathbf{Z}_{\mathrm{prompt}}(x,t) \right).
\label{eq:llm_asr_generation}
\end{equation}

\noindent\textbf{Entropy allocation at the encoder--LLM interface.}
The encoder--LLM interface is constituted by the speech embedding $Z$ that the LLM directly consumes. Since $\mathcal{A}_\psi$ is a lightweight MLP, the structural properties of $Z$ are largely inherited from $E$, which we therefore analyze as the encoder-side proxy of the interface. Given a fixed transcription target per utterance, the encoder and LLM face a complementary trade-off in resolving uncertainty: what the encoder absorbs upstream reduces what the LLM must resolve downstream, making $E$ a principled lens on cross-module entropy allocation.

%The representation $E$ defined above serves as the interface between the encoder and the LLM. Since the total uncertainty to be resolved for a given utterance is invariant, the two modules can be viewed as operating under a zero-sum entropy budget: uncertainty absorbed by the encoder directly reduces what the LLM must resolve. Analyzing the properties of $E$ thus provides a principled lens into the entropy allocation across modules.

\subsection{Metrics on Encoder Representations}
\label{sec3.2}

%Building on this entropy-based perspective, we characterize the encoder--LLM interface through the encoder representation $E$. 
For each utterance, let $E'\in \mathbb{R}^{L_{valid} \times d_e}$ denote the valid-frame representation obtained from $E$ after removing padded positions according to the sequence mask. We characterize $E'$ from two complementary aspects: the overall spectral entropy retained in representations, and proxy estimates of how uncertainty reduction remains accessible in phonetic and semantic target spaces.

\noindent\textbf{Normalized spectral entropy (NSE).} 
We perform singular value decomposition (SVD) on the encoder representation matrix $E' = U \Sigma V^\top$, where $\Sigma = \mathrm{diag}(\sigma)$, with $\sigma=(\sigma_1,\ldots,\sigma_d)^\top$ being singular values in descending order and $d=\min(L_{valid},d_e)$. Treating the normalized singular values as a distribution~\citep{roy2007effective}, we define normalized spectral entropy~\citep{yang2005coefficient} as
\begin{equation}
\mathrm{NSE}(E')
=
-\frac{1}{\log d}
\sum_{i=1}^{d}
\bar{\sigma}_i \log \bar{\sigma}_i,
\quad
\bar{\sigma}_i=\frac{\sigma_i}{\sum_{j=1}^{d}\sigma_j},
\label{eq:svd_entropy}
\end{equation}
which captures the global spectral geometry of $E'$. Unlike multi-task regimes where isotropy is typically desirable, LLM-ASR has a single objective; here, high NSE reflects an isotropic representation space with excess entropy, whereas low NSE reflects a more anisotropic, compressed one.

\noindent\textbf{Phonetic and conditional semantic accessible information.} 
While NSE quantifies the global spectral entropy retained in $E'$, it remains agnostic to whether the entropy aligns with task-relevant subspaces. We therefore introduce two complementary probe-inspired accessible-information proxies: \textbf{phonetic accessible information (PAI)}, measuring how much phonetic information is linearly accessible from $E'$; and
\textbf{conditional semantic accessible information (CSAI)}, measuring how much residual semantic information is accessible beyond what the phonetic target space linearly explains. The dynamics of PAI and CSAI can reflect how the representation space evolves during training.

To derive them, we first summarize the valid-frame representation $E'$ via temporal mean and standard deviation pooling, followed by Principal Component Analysis (PCA) and standardization, yielding a vector $u$. For the phonetic target $P$, we convert the reference transcript into phoneme tokens, form an $\ell_1$-normalized bag-of-phones count vector, and apply the same PCA-standardization pipeline. For the semantic target $S$, we encode the transcript using a frozen Qwen3-Embedding-8B~\citep{yang2025qwen3} model with last-token pooling and $\ell_2$ normalization, followed by PCA and standardization. In practice, $P$ and $S$ are projected into the same dimension. 
% Since both $P$ and $S$ are derived from transcripts, PAI and CSAI do not cleanly disentangle acoustic evidence from all transcript-derived structure. We use them only as relative, within-protocol probes: PAI emphasizes accessibility to a phonetic bag-of-phones target, whereas CSAI emphasizes residual dependence on text-embedding structure after conditioning on that phonetic proxy.

Let $q=[u^\top,P^\top,S^\top]^\top$ and let $\hat{\Sigma}=\mathrm{Cov}(q)$ denote the empirical covariance estimated over the evaluation set. We use a ridge-regularized covariance
$\tilde{\Sigma}=\hat{\Sigma}+\lambda I$, where $\lambda>0$ ensures numerical stability. All covariance blocks below are taken as principal submatrices of $\tilde{\Sigma}$. Under a joint Gaussian approximation on $(u,P,S)$, we define
\begin{equation}
\scalebox{0.92}{$\displaystyle
  \mathrm{PAI}(E')
  =
  \left[
    \frac{1}{2\log 2}
    \log
    \frac{\det \tilde{\Sigma}_{uu}\,\det \tilde{\Sigma}_{PP}}
         {\det \tilde{\Sigma}_{[u,P]}}
  \right]_+
$}
\label{eq:pai}
\end{equation}
where $[\cdot]_+=\max(\cdot,0)$ clips residual negative values caused by numerical error. This quantity serves as a regularized Gaussian accessible-information proxy for the mutual information between $u$ and $P$. 
Similarly, letting $\tilde{\Sigma}_{\cdot\mid P}$ denote the corresponding regularized conditional covariance matrices given $P$, computed from the same joint covariance $\tilde{\Sigma}$, we define
\begin{equation}
\scalebox{0.92}{$\displaystyle
  \mathrm{CSAI}(E')
  =
  \left[
    \frac{1}{2\log 2}
    \log
    \frac{\det \tilde{\Sigma}_{uu\mid P}\,\det \tilde{\Sigma}_{SS\mid P}}
         {\det \tilde{\Sigma}_{[u,S]\mid P}}
  \right]_+
$}
\label{eq:csai}
\end{equation}
which serves as a regularized Gaussian proxy for the conditional accessible-information between $u$ and $S$ given $P$. Intuitively, CSAI captures the semantic information in $u$ unexplained by the bag-of-phones target, by removing the portion of $u$--$S$ dependence mediated by the phonetic proxy. It is worth noting that PAI and CSAI, derived under a Gaussian approximation with transcript-based targets, serve as relative within-protocol proxies rather than exact mutual-information estimates.

Together, NSE, PAI, and CSAI form a practical diagnostic: NSE quantifies the spectral entropy, while PAI and CSAI quantify phonetic accessibility and residual semantic accessibility. %PAI/CSAI are relative proxies, not exact mutual-information estimates. 
Jointly analyzing these metrics helps diagnose entropy-allocation dynamics across model components. As discussed in Section~\ref{sec1} regarding Figure~\ref{fig1}, the trends of these metrics can reveal suboptimal behavior during training. In a desirable training trajectory, we expect NSE not to remain persistently high, thereby reducing the capacity demand on the LLM; meanwhile, we expect CSAI not to exhibit a sustained rise at the expense of declining PAI during joint training, as this pattern indicates that the improvement stems not from representation refinement but from linguistic shortcuts acquired at the cost of acoustic fidelity---a trend that can amplify hallucination risk.

\subsection{Design Principle}
\label{sec3.3}

With these metrics, we can provide an intuitive diagnosis of entropy allocation in LLM-ASR models. Figure~\ref{fig1} exposes two contrasting suboptimal modes, which have been analyzed in detail in Section~\ref{sec1}. We therefore adopt a \textbf{capability-boundary-aware} design principle with two requirements:
%\begin{itemize}%[nosep,leftmargin=*]
\begin{itemize}[leftmargin=1.6em, itemsep=2pt, parsep=0pt, topsep=1pt]
\item The encoder should be guided toward low-entropy, acoustically grounded representations before being exposed to LLM-dominated joint optimization. This narrows the modality gap early and reduces the risk that subsequent training resorts to semantically biased shortcuts.
\item During joint optimization, the functional boundary between modules should be explicitly maintained, ensuring that further encoder compression proceeds along an acoustically grounded direction rather than at the expense of acoustic representation quality.
\end{itemize}
% The first principle motivates our pretraining strategy, while the second motivates the IA-SFT stage.

\subsection{Multi-Stage Training Paradigm}
\label{sec3.4}

As shown in Figure~\ref{fig_sft}, we illustrate our proposed training pipeline alongside a comparison with the traditional one. Our design features two core components: phoneme-level CTC pretraining, and an additional IA-SFT stage that runs asynchronously in parallel with the pretraining stage.

\noindent\textbf{Phoneme-level CTC Pretraining.}
The pretraining stage starts from a Conformer encoder initialized with FireRedASR-AED weights, with the original autoregressive decoder replaced by a lightweight linear CTC head trained under the CR-CTC~\citep{yao2024cr} objective. Motivated by prior findings that phoneme-based representations can offer a more universal and acoustically grounded interface than grapheme or subword units~\citep{yusuyin2025whistle}, we adopt phoneme-level rather than character- or subword-level supervision.

It is worth noting that we intentionally adopt CTC-based pretraining rather than the more prevalent AED-style or self-supervised alternatives. This choice is driven by two structural considerations. First, the CTC objective, through its peaky behavior~\citep{zeyer2021does} and monotonic alignment constraint, encourages the encoder to compress continuous speech into representations more aligned with the underlying token sequence~\citep{zhou2025cjst}, bridging the structural gap with discrete text. Second, the lightweight CTC head acts as a capacity bottleneck, pushing the encoder toward low-entropy representations. Together, these properties narrow the speech--text modality gap at the representation level, while the lower-entropy encoder representation relaxes the capacity demand on the downstream LLM through the interface.

\begin{figure}[!htbp]
    \centering
    \includegraphics[width=0.88\linewidth]{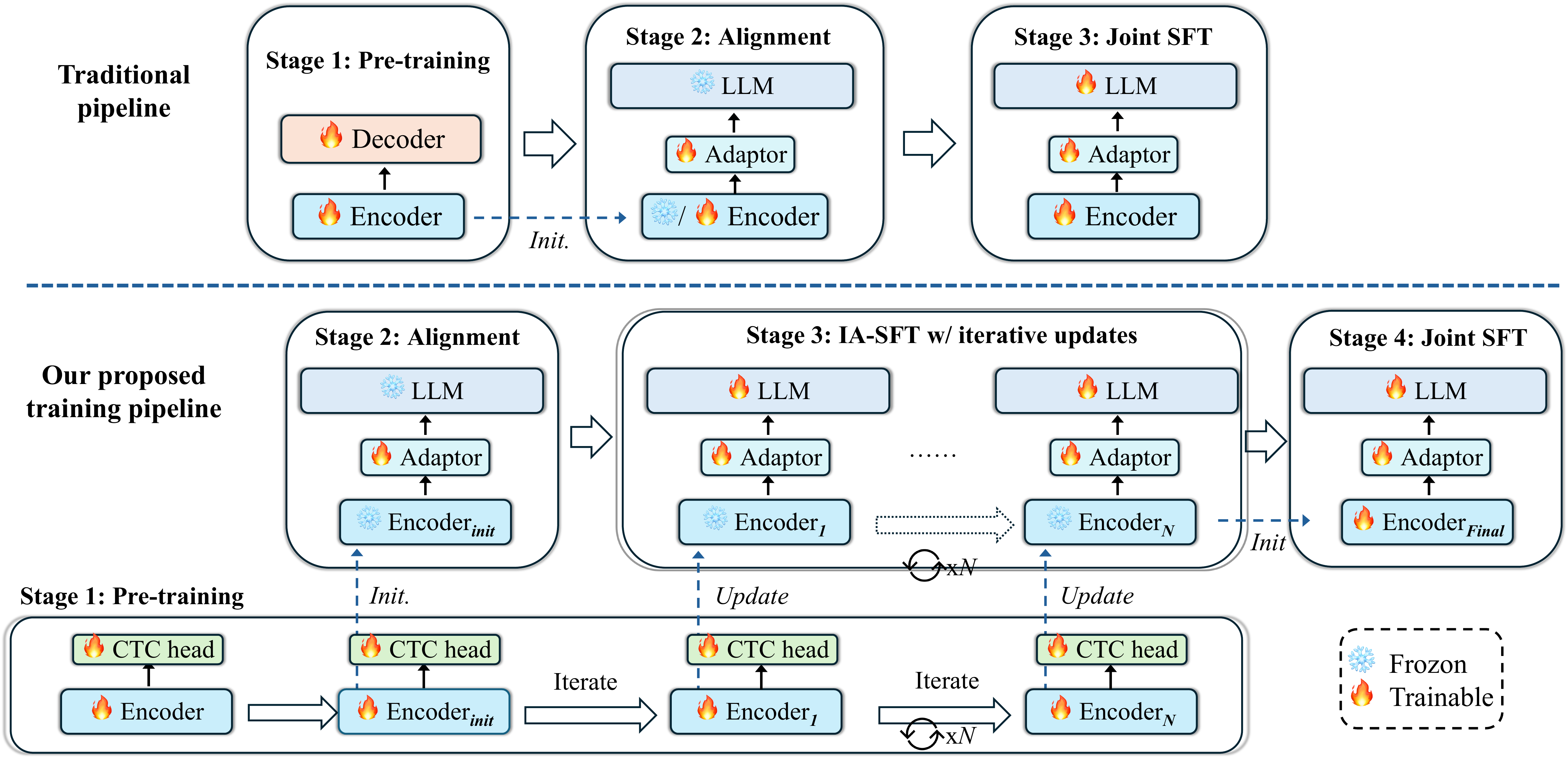}
    \caption{Comparison of our multi-stage training design with the
traditional training pipeline.}
    \label{fig_sft}
    \vskip -5mm
\end{figure}

\noindent\textbf{Alignment and IA-SFT.}
In traditional training paradigms, alignment and joint SFT are conducted sequentially after pretraining completes. To explicitly constrain representation drift, we introduce an additional IA-SFT stage between alignment and joint SFT. As shown in Figure~\ref{fig_sft}, neither alignment nor IA-SFT waits for pretraining to complete; instead, they proceed asynchronously in parallel with pretraining. To govern encoder initialization and updates for the alignment and IA-SFT stages, we monitor representation drift using Centered Kernel Alignment (CKA)~\citep{kornblith2019similarity}, which is applied to compare the evolving encoder with a reference checkpoint that is initialized and updated throughout pretraining. The detailed procedure is as follows:

\textbf{(1) Alignment:} Once pretraining has run for a sufficient number of steps, we snapshot the encoder to initialize the reference checkpoint. As pretraining continues, we monitor the CKA score between the evolving encoder and reference checkpoint. When the CKA score first drops below a preset threshold, we take a new snapshot that both initializes the speech encoder for alignment and updates the reference checkpoint. During alignment, both the encoder and LLM are frozen, and only the adaptor is trained;

\textbf{(2) IA-SFT:} After alignment completes, the model enters the IA-SFT stage, where the encoder is frozen while the adaptor and LLM are trained. In essence, IA-SFT also serves an alignment purpose---it can be viewed as a preparatory curriculum that strengthens the LLM's capacity to comprehend speech representations before joint SFT. To expose the LLM to diverse encoder representations, we perform CKA-driven iterative encoder hot-swapping: whenever the CKA score drops below the threshold, the latest encoder checkpoint from the ongoing pretraining is hot-swapped into the frozen encoder used in IA-SFT, and the reference checkpoint is updated accordingly. The CKA constraint ensures that each swap delivers a substantive shift in encoder representations, enabling the LLM to learn from progressively refined encoders in a curriculum-like fashion. Throughout the IA-SFT stage, this iterative hot-swapping cycle repeats until pretraining concludes. Further details of IA-SFT are supplemented in Appendix~\ref{appendA4}, including a more detailed procedural description, the formal definition of CKA, specific training configurations, and an analysis of why direct encoder hot-swapping without realignment is viable.

\noindent\textbf{Joint SFT.}
After IA-SFT, we proceed to joint SFT, where all modules are trained end-to-end. By this point, the modality gap has been substantially narrowed, and the speech representations can preserve their acoustic-phonetic structure. This enables the LLM to reliably interpret speech-derived representations, thereby mitigating the risk of representation drift during joint training. Notably, although an additional stage is introduced, both alignment and IA-SFT run asynchronously in parallel with pretraining from its midpoint, ensuring that our overall pipeline remains time-efficient in large-scale industrial settings.

\subsection{Empirical Analysis of Metric Dynamics}
\label{Metric_Empirical}

To empirically examine whether our multi-stage training paradigm adheres to the design principles, we trace the proposed metrics across training stages and compare them against the direct transfer path from FireRedASR-AED to FireRedASR-LLM, as illustrated in Figure~\ref{fig_svd}. The alignment stage is omitted from the figure, as no metric change occurs in this encoder-frozen stage. Together, the metric trajectories reveal not only how much entropy reduction the encoder assumes across stages, but also whether such compression remains acoustically grounded or drifts toward semantic bias.

\begin{figure}[!htbp]
    \centering
    \includegraphics[width=0.9\linewidth]{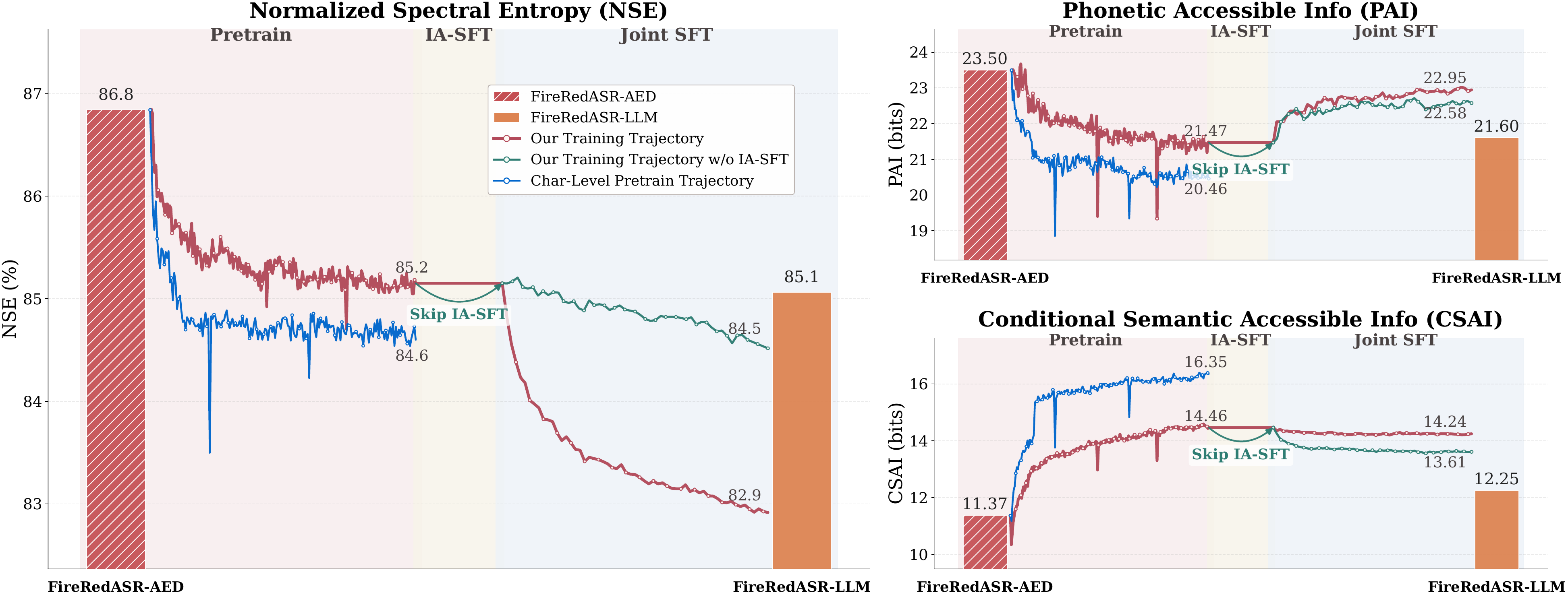}
    % \vskip -1mm
    \caption{Comparison of the three metrics along our training trajectory and the FireRedASR-AED $\rightarrow$ FireRedASR-LLM transition, sharing the same encoder architecture. ``Our Training Trajectory'' corresponds to the full standard ``phoneme-level pretraining $\rightarrow$ IA-SFT $\rightarrow$ joint SFT'' pipeline.}
    \label{fig_svd}
    \vskip -2mm
\end{figure}

During phoneme-level CTC pretraining, NSE drops noticeably, as the CTC bottleneck compels the encoder to form compact, low-entropy representations. Meanwhile, PAI decreases while CSAI rises steadily. Given that only phoneme-level supervision is provided, the PAI decline does not indicate acoustic information loss; rather, it is a structural consequence of the peaky, sparse output distributions induced by CTC. Such distributions bias representations toward discrete, symbol-like forms, which dilute the globally accessible linear phonetic structure after pooling and PCA. The moderate CSAI increase follows naturally, as monotonic alignment progressively drives representations toward finer token-level granularity. Moreover, we present a comparison between phoneme-level and character-level pretraining variants. Phoneme-level pretraining consistently achieves higher PAI and lower CSAI than its character-level counterpart throughout training, suggesting an advantage in instilling a robust acoustic foundation while suppressing premature semantic anchoring in the encoder.

During joint SFT, NSE continues to decrease along our training trajectory, indicating that end-to-end optimization further drives the encoder toward lower-entropy representations and progressively relieves downstream capacity pressure on the LLM. Furthermore, a key pattern emerges in our trajectory: PAI recovers notably while CSAI remains flat or declines slightly---consistent with our design expectations. Since joint optimization begins from an already-narrowed modality gap and a well-aligned interface, it mitigates representation drift and instead promotes a desirable cross-modal division of labor. Specifically, LLM gradients incentivize the encoder to refine phonetic-acoustic representations rather than exploit linguistic shortcuts, while the encoder refrains from absorbing semantic functions better delegated to the LLM. This emergent modular separation reflects efficient parameter utilization across the architecture.

To quantify the contribution of IA-SFT, we present an ablation in which this stage is removed (green curves in Figure~\ref{fig_svd}). Without IA-SFT, NSE decays more slowly relative to our full training trajectory, while both PAI and CSAI converge to lower final values---confirming degraded linear accessibility along both acoustic and semantic dimensions, or equivalently, a reduction in the effective signal-to-noise ratio of encoder representations. These results underscore the critical role of IA-SFT as a representational buffer: by progressively deepening the alignment between the LLM's embedding manifold and encoder representations, it both preserves the acoustic grounding established during pretraining and shields against representation drift in the subsequent joint optimization stage.

\section{Experiments}
\label{sec4}

\subsection{Evaluation setting}
We evaluate on public ASR benchmarks covering Mandarin~\citep{bu2017aishell,du2018aishell}, Chinese dialects~\citep{tang2021kespeech,dai2025wenetspeech,li2026wenetspeech}, English~\citep{panayotov2015librispeech,wang2021voxpopuli}, and Chinese--English code-switching~\citep{lovenia2022ascend,zhou2025cs}. We report Character Error Rate (CER) for Chinese and Word Error Rate (WER) for English. Beyond LLM-ASR baselines, we include two large-scale LALMs~\citep{xu2025qwen3,wu2025step} as references. Despite their substantially higher latency and overhead, they help assess competitiveness beyond lightweight settings. All transcriptions are normalized with the \textit{WeTextProcessing} toolkit. All models are evaluated in offline mode with their official decoding configurations, and our model uses beam search with a beam size of 3.

\begin{table}[!htbp]
\centering
\caption{
Comparison with advanced LLM-ASR and LALM baselines on public benchmarks. %``-'' denotes unsupported dialects.
}
\label{tab:main_benchmark}
\resizebox{\linewidth}{!}{
%\renewcommand{\arraystretch}{0.92}
%{\small
\begin{tabular}{lcccccc|c}
\toprule
  & \textbf{Fun-ASR-nano} & \textbf{GLM-ASR-nano} & \textbf{Qwen3-ASR-1.7B} & \textbf{FireRedASR-LLM} & \textbf{Step-Audio2-mini} & \textbf{Qwen3-Omni-Inst} & \textbf{Ours} \\
  % & \textbf{Fun-ASR} & \textbf{GLM-ASR} & \textbf{Qwen3-ASR} & \textbf{FireRedASR} & \textbf{Step-Audio2} & \textbf{Qwen3-Omni} & \textbf{Ours} \\
  % & \textbf{nano} & \textbf{nano} & \textbf{1.7B} & \textbf{LLM} & \textbf{mini} & \textbf{Instruct} &  \\
  
  \midrule
  Model Size & 0.8B & 1.5B  & 2.0B  & 8B+   & 8B+   & 30B-A3B  & 2.3B \\
  \midrule
\multicolumn{8}{l}{\textit{\textbf{Mandarin}}} \\
AISHELL-1 \textit{dev $\mid$ test} & 1.59 $\mid$ 1.81 & 2.40 $\mid$ 2.41 & 1.40 $\mid$ 1.51 & 0.71 $\mid$ 0.73 & 0.76 $\mid$ 0.81 & 0.86 $\mid$ 0.92 & \textbf{0.45} $\mid$ \textbf{0.59} \\
AISHELL-2-ios \textit{dev $\mid$ test} & 2.62 $\mid$ 2.73 & 3.21 $\mid$ 3.45 & 2.41 $\mid$ 2.60 & \textbf{2.08} $\mid$ \textbf{2.12} & 2.24 $\mid$ 2.29 & 2.11 $\mid$ 2.31 & 2.32 $\mid$ 2.45 \\
AISHELL-2021-\textit{Eval} \textit{A $\mid$ C $\mid$ D} & 4.75 $\mid$ 4.29 $\mid$ 2.33 &7.25 $\mid$ 9.48 $\mid$ 3.40  &4.22 $\mid$ 3.51 $\mid$ 1.82 & 12.61 $\mid$ 4.06 $\mid$ 7.38 &4.54 $\mid$ 3.69 $\mid$ 2.34 & 5.19 $\mid$ 3.34 $\mid$ \textbf{1.66} &  \textbf{3.45} $\mid$ \textbf{1.71} $\mid$ 2.53   \\

% WeNetSpeech \textit{meeting $\mid$ net} & 4.68 $\mid$ 5.22 & 6.87 $\mid$ 5.72  & 4.00 $\mid$ 4.13 & 3.50 $\mid$ 3.80 & 4.23 $\mid$ 4.63 & 3.92 $\mid$ 3.85 & 4.90 $\mid$ 4.71 \\
% SpeechIO  & 2.78 & 3.17  & 2.55 & 2.28 & 3.41  & 2.33 & 2.62 \\
\midrule
\multicolumn{8}{l}{\textit{\textbf{Chinese Dialect}}} \\
WeNetSpeech-Chuan \textit{easy $\mid$ hard}  & 12.69 $\mid$ 23.76 & 20.95 $\mid$ 33.61 & 11.40 $\mid$ \textbf{20.35} & 12.14 $\mid$ 24.76 & 13.99 $\mid$ 25.35 & 14.13 $\mid$ 25.16 & \textbf{10.94} $\mid$ 21.93 \\
WeNetSpeech-Yue \textit{short $\mid$ long}  & 7.31 $\mid$ 10.02 & 16.78 $\mid$ 13.97 & 5.79 $\mid$ \textbf{8.00} & - $\mid$ - & 7.78 $\mid$ 8.44 & 6.97 $\mid$ 8.60  & \textbf{5.22} $\mid$ 9.45 \\
KeSpeech  & 7.18 & 9.59  & 4.98 & \textbf{3.53} & 3.98 & 6.00 & 4.56 \\
  
\midrule
\multicolumn{8}{l}{\textit{\textbf{English}}} \\
LibriSpeech-dev \textit{clean $\mid$ other} & 1.63 $\mid$ 4.06 & 1.82 $\mid$ 3.93 & 1.54 $\mid$ 3.14 & 1.25 $\mid$ 2.92 & \textbf{1.06} $\mid$ 2.48 & 1.08 $\mid$ \textbf{2.10} & 1.11 $\mid$ 2.57 \\
LibriSpeech-test \textit{clean $\mid$ other} & 1.63 $\mid$ 4.35 & 1.96 $\mid$ 4.29 & 1.56 $\mid$ 3.49 & 1.37 $\mid$ 3.36 & 1.22 $\mid$ 2.61 & \textbf{1.15} $\mid$ \textbf{2.38} & 1.23 $\mid$ 2.63 \\
VoxPopuli \textit{dev $\mid$ test} & 7.86 $\mid$ 7.70 & 8.78 $\mid$ 8.52 & 7.58 $\mid$ 7.42 & 10.65 $\mid$ 10.26 & 8.86 $\mid$ 8.37 & 6.86 $\mid$ 6.75 & \textbf{6.25} $\mid$ \textbf{6.22} \\
% TED-LIUM & 3.66 & 3.18  & 2.70 & 3.89 & 3.35 & 3.07 & 3.61 \\
% MLS-English & 6.80 & 5.32  & 4.93 & 5.32 & 4.37 & 4.04 & 4.85 \\
 
\midrule
\multicolumn{8}{l}{\textit{\textbf{Chinese--English Code-switch}}} \\
CS-Dialogue & 5.37 & 6.15  & 5.44 & 5.10 & 9.46 & 8.51  & \textbf{4.99} \\
ASCEND  & 11.91 & 12.29  & \textbf{10.87} & 11.25 & 13.50 & 18.68 & 11.79 \\
\midrule%\midrule
% \rowcolor{black!6}
% \multicolumn{8}{l}{\textbf{Ours vs. Baselines}} \\
% \addlinespace[1pt]
% \rowcolor{black!6}
\textbf{Avg. CER/WER}  & \textbf{6.28} & \textbf{8.71}  & \textbf{5.45} & \textbf{6.46} & \textbf{6.19} & \textbf{6.24} & \textbf{5.12}  \\
\bottomrule
\end{tabular}
%}
}
% \vskip -1mm
\end{table}

\subsection{Main Recognition Results}

Table~\ref{tab:main_benchmark} compares our model with advanced ASR baselines~\citep{an2025funaudio,zai2025glmasrnano,shi2026qwen3,xu2025fireredasr,wu2025step,xu2025qwen3}. With only 2.3B parameters, our model achieves competitive results across multilingual benchmarks, surpassing several industrial-scale models with over 8B parameters on multiple scenarios. Beyond classic benchmarks such as AISHELL and LibriSpeech, our model attains SOTA performance on entity-dense sets like AISHELL-2021-Eval (in-car, telephony), suggesting that aligning the LLM to low-entropy speech representations does not cause catastrophic forgetting of world knowledge. Furthermore, our model achieves leading results on several dialect benchmarks, indicating strong robustness to phonetic shifts induced by pronunciation variation. This behavior is supported by the elevated PAI values in Figure~\ref{fig_svd}, validating our acoustically grounded design principle. Our model also performs well on code-switching benchmarks, owing in part to the phoneme-level pretraining strategy, which leverages language-agnostic phonetic representations shared across languages, thereby enabling more robust modeling under code-switching conditions.

\subsection{Hallucination Mitigation}

\begin{wraptable}{r}{0.58\columnwidth}
\vspace{-1.0em}
\caption{Hallucination rates on scenario benchmarks.}
\label{tab:hallucination_rate}
\vspace{-0.5em}
\resizebox{\linewidth}{!}{%
\renewcommand{\arraystretch}{0.92}
\begin{tabular}{lcccc}
\toprule
Model & Mandarin & Dialect & English & Code-switch \\
\midrule
Fun-ASR-nano      & 0.018\% & 0.217\% & 0.014\% & 0.397\% \\
GLM-ASR-nano      & 0.030\% & 0.201\% & 0.014\% & 0.315\% \\
Qwen3-ASR-1.7B    & 0.018\% & \textbf{0.120\%} & 0.014\% & 0.345\% \\
FireRedASR-LLM    & 0.053\% & 0.228\% & 0.014\% & 0.324\% \\
Step-Audio2-mini  & 0.020\% & 0.194\% & 0.014\% & 1.255\% \\
Qwen3-Omni-Inst   & 0.013\% & 0.370\% & \textbf{0.007\%} & 1.778\% \\
\midrule
Ours (w/o IA-SFT) & 0.005\% & 0.198\% & 0.014\% & 0.356\% \\
Ours              & \textbf{0.003\%} & 0.122\% & \textbf{0.007\%} & \textbf{0.261\%} \\
\bottomrule
\end{tabular}%
}
\vspace{-1.0em}
\end{wraptable}

A core contribution of our method is its principled approach to hallucination mitigation.
To validate this, we use the same models and benchmarks from Table~\ref{tab:main_benchmark}
and report the average hallucination rate per scenario, defined as the ratio of hallucinated
samples to total samples across all benchmarks within each scenario.
Specifically, a sample is classified as hallucinated if its transcription is more than 50\% longer than the ground truth and has no unigram overlap with it, i.e., a BLEU-1 score of zero.
As shown in Table~\ref{tab:hallucination_rate}, our model achieves substantially lower hallucination rates than all baselines. This suggests that our design mitigates hallucination at a fundamental level, owing to the explicit constraint during training that enforces the encoder to remain acoustically grounded, thereby suppressing representation drift. We further ablate the IA-SFT module and identify it as a major contributor to our low hallucination rate. Without IA-SFT, even low-entropy representations are susceptible to being influenced by the LLM's gradients, causing the model to fall back on linguistic shortcuts and produce hallucinations---a failure mode that IA-SFT can help mitigate.

\subsection{Layer-wise Alignment: From Encoder to Adaptor}

While the preceding analysis focuses on encoder representations, we further examine post-adaptor embeddings. Figure~\ref{fig:cka_layerwise} visualizes the progression of semantic alignment by computing linear CKA scores between ground-truth text embeddings and representations at each encoder layer as well as the post-adaptor embedding. Linear CKA is the special case of the general CKA defined in Eq.~\eqref{eq:cka}, measuring alignment across features of different dimensionalities via centered Gram matrices that capture inter-sample similarity structures.
To handle sequence length mismatches, we apply temporal mean pooling to obtain utterance-level representations $X_l \in \mathbb{R}^{B\times d_e}$ or $Z \in \mathbb{R}^{B\times d_m}$ from the $l$-th layer ($B$ denotes the number of utterances), with the same pooling applied to ground-truth text embeddings $Y \in \mathbb{R}^{B\times d_m}$:
\vspace{-2pt}
\begin{equation}
\mathrm{CKA}_{\text{linear}}(X_l, Y)
=
\frac{\|X_l^\top Y\|_F^2}
     {\|X_l^\top X_l\|_F \, \|Y^\top Y\|_F},
\label{eq:cka_linear}
\end{equation}
\vspace{-2pt}
\begin{wrapfigure}{r}{0.55\columnwidth}
    \vspace{-1.0em}
    \centering
    \includegraphics[width=\linewidth]{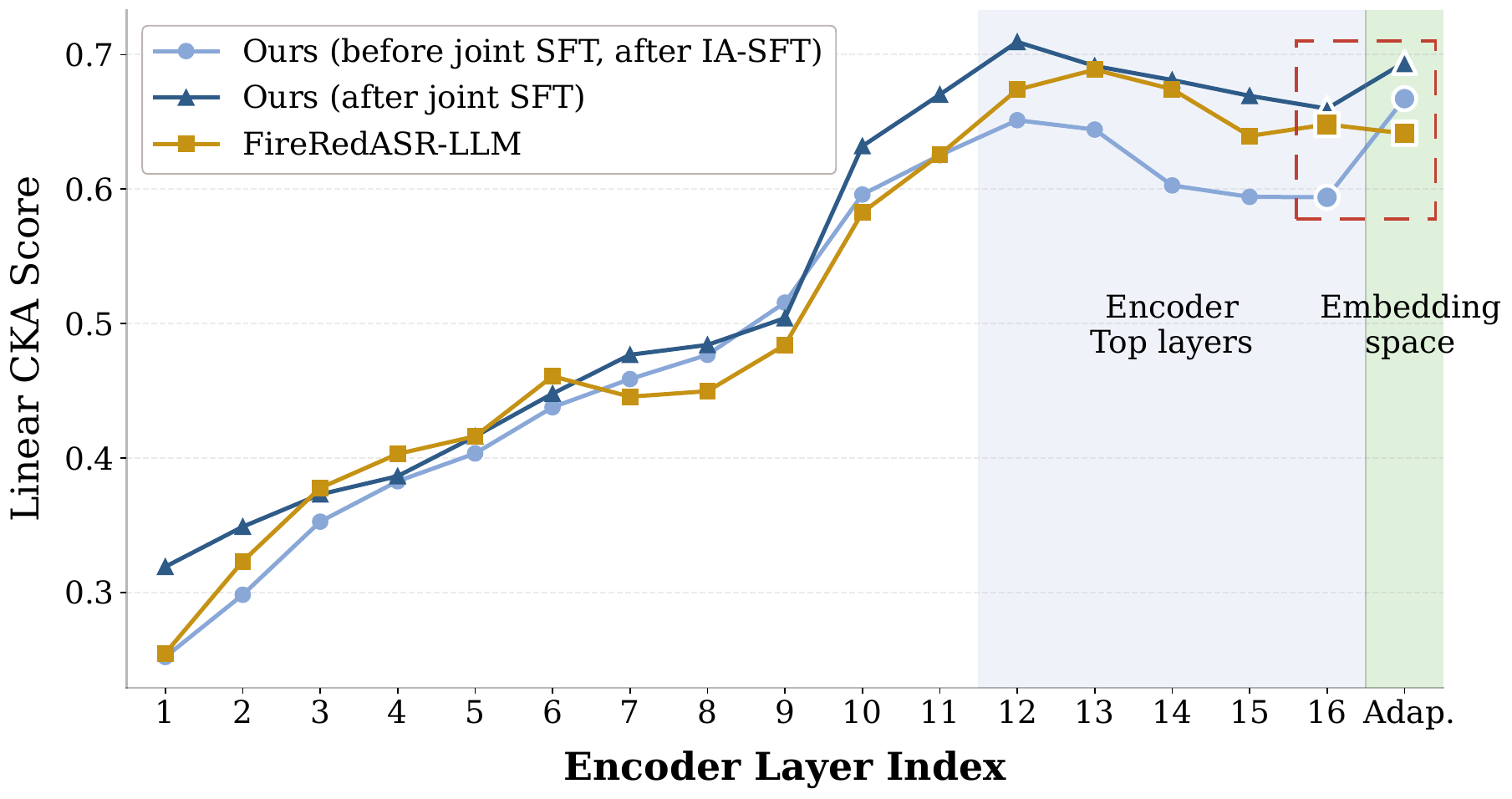}
    \vspace{-0.8em}
    \caption{
        Linear CKA scores between layer-wise representations and ground-truth text embeddings, averaged over 1,000 AISHELL/LibriSpeech utterances; ``Adap.'' denotes the post-adaptor embedding.
    }
    \label{fig:cka_layerwise}
    \vspace{-0.4em}
\end{wrapfigure}
where $\|\cdot\|_F$ denotes the Frobenius norm. We further include FireRedASR-LLM as a reference, whose encoder and adaptor share the same architectures as ours, enabling a controlled comparison. As shown in Figure~\ref{fig:cka_layerwise}, our model and FireRedASR-LLM exhibit markedly different behaviors at the adaptor interface. For FireRedASR-LLM, the CKA score decreases after adaptor projection, indicating that its high-entropy encoder representations are difficult to align with the text embedding space. The adaptor has to simultaneously perform dimensional projection, stabilize the interface, and reconcile semantic geometry---a compound burden that exceeds its limited capacity. In contrast, our model consistently shows a pronounced CKA increase at this interface, regardless of whether joint SFT has been applied. Since our encoder already produces low-entropy representations roughly synchronized with text tokens, and IA-SFT further familiarizes the LLM with these representations, the adaptor is largely relieved of compensatory duties. It needs only to perform a straightforward dimensional mapping while correcting residual geometric misalignment, yielding embeddings structurally better matched to the text manifold.

\subsection{Ablation Study}

\begin{table}[!htbp]
    \centering
    %\caption{Ablation of post-training strategies; removing encoder iter.\ in IA-SFT disables hot-swapping and asynchronous updates.}
    \caption{Ablation study on post-training strategies. ``– Encoder iter. in IA-SFT'' denotes removing encoder hot-swapping and asynchronous parallel during IA-SFT, reducing it to a standard encoder-frozen stage with only the adaptor and LLM trained.}
    \label{tab:ablation_strategy}
    
    \resizebox{0.7\linewidth}{!}{
    \renewcommand{\arraystretch}{1.0}
    \begin{tabular}{l c c c c}
        \toprule
        % \textbf{Configuration} & \multicolumn{4}{c}{\textbf{Avg.\ CER\% / WER\%}} \\
        % \cmidrule(lr){2-5}
        \textbf{Configuration} & \textbf{Mandarin} & \textbf{Dialect} & \textbf{English} & \textbf{Code-switch} \\
        \midrule
        Our full pipeline & \textbf{1.93} & \textbf{10.42} & \textbf{3.34} & \textbf{8.39} \\
        \quad -- Joint SFT & 2.18 & 12.84 & 4.22 & 10.15 \\
        \quad -- IA-SFT & 2.08 & 11.47 & 3.79 & 9.11 \\
        \quad -- Encoder iter.\ in IA-SFT & 1.95 & 10.87 & 3.40 & 8.57 \\
        \bottomrule
    \end{tabular}
    }
    % \vskip -3mm
\end{table}

We conduct ablation studies on post-training strategies under controlled conditions, using identical training data and configurations across all experiments, each trained until validation loss plateaus for three consecutive checkpoints (10k-step intervals). As shown in Table~\ref{tab:ablation_strategy}, ablating joint SFT yields the largest degradation, confirming that end-to-end optimization is essential for refining the representation space. Moreover, ablating IA-SFT also causes substantial performance drops, consistent with the analysis in Section~\ref{Metric_Empirical} and Table~\ref{tab:hallucination_rate}: without the drift constraint imposed by IA-SFT, speech representations tend to shift toward the semantic subspace, undermining phoneme discrimination and amplifying hallucination. We further ablate the encoder hot-swapping mechanism within IA-SFT, reducing it to a standard encoder-frozen stage where only the adaptor and LLM are trained using the final encoder checkpoint. This variant still underperforms our full IA-SFT, since static representations from a frozen encoder offer insufficient diversity for robust adaptation, which may make the LLM more prone to suboptimal patterns within a narrow phonetic subspace. In contrast, encoder hot-swapping mitigates this by acting as implicit regularization: exposing the LLM to progressively evolving encoder states encourages learning of robust patterns shared across encoders, rather than those specific to any single encoder.

\section{Conclusion and Limitations}
\label{sec6}

In this work, we revisit LLM-based ASR from an entropy-allocation perspective and propose a capability-boundary-aware framework that encourages a clearer division of labor between the encoder and LLM for acoustic uncertainty reduction and semantic ambiguity resolution. Through phoneme-level CTC pretraining and IA-SFT, our approach improves module-wise specialization and achieves strong performance with low hallucination rates across diverse benchmarks using only 2.3B parameters. 
This study also has several limitations that naturally point to future work, including scaling the framework to larger LLM backbones, extending the entropy-allocation analysis to RL, and developing stronger diagnostics beyond the current linear-Gaussian proxies. A detailed discussion of these limitations and corresponding future directions is provided in Appendix~\ref{sec:appendix_limitations}.

\begin{ack}
We thank our colleagues in the LLM Team of the
Advanced Intelligence Systems Group at NIO
for their valuable discussions and support.

\textbf{Funding and Competing Interests.}
This research was conducted at NIO.
The authors declare no additional external funding
or competing interests beyond their disclosed
affiliation with NIO.
\end{ack}

\bibliography{example_paper}
\bibliographystyle{plainnat}

%%%%%%%%%%%%%%%%%%%%%%%%%%%%%%%%%%%%%%%%%%%%%%%%%%%%%%%%%%%%%%%%%%%%%%%%%%%%%%%
%%%%%%%%%%%%%%%%%%%%%%%%%%%%%%%%%%%%%%%%%%%%%%%%%%%%%%%%%%%%%%%%%%%%%%%%%%%%%%%
% APPENDIX
%%%%%%%%%%%%%%%%%%%%%%%%%%%%%%%%%%%%%%%%%%%%%%%%%%%%%%%%%%%%%%%%%%%%%%%%%%%%%%%
%%%%%%%%%%%%%%%%%%%%%%%%%%%%%%%%%%%%%%%%%%%%%%%%%%%%%%%%%%%%%%%%%%%%%%%%%%%%%%%

\clearpage
\appendix

\section{Detailed Limitations and Future Work}
\label{sec:appendix_limitations}

While our entropy-allocation framework and the proposed multi-stage training paradigm provide a principled account of the encoder--LLM division of labor in LLM-based ASR, several limitations remain and naturally point to directions for future work.

\paragraph{Scalability to Larger Backbones.}
This work targets efficient and stable LLM-ASR models and demonstrates strong effectiveness in lightweight deployment scenarios, but the proposed framework is not inherently confined to small-scale models. Due to time and resource constraints, we have not yet explored larger backbones. Future work will instantiate the framework with larger yet inference-efficient backbones, such as MoE-based Qwen3-30B-A3B, to assess whether the capability-boundary-aware design and IA-SFT paradigm generalize to medium-scale models and further improve the performance upper bound of LLM-ASR systems.

\paragraph{Limited training paradigm coverage.}
This study focuses on encoder pretraining, alignment, and SFT, and does not yet cover reinforcement learning. Given its growing adoption in recent LLM-ASR systems such as Seed-ASR and Fun-ASR, an important future direction is to examine how RL reshapes entropy allocation and whether capability-boundary-aware constraints can be incorporated into reward design or regularization to improve robustness and reduce hallucinations.

\paragraph{Diagnostic Metrics as Relative Proxies.}
The PAI and CSAI metrics introduced in Section~\ref{sec3.2} are based on a joint Gaussian approximation with ridge-regularized covariance estimation, and are designed as linear-Gaussian accessible-information proxies rather than exact mutual-information estimates. They offer a practical and consistent diagnostic tool for relative comparisons across training paradigms in our setting, while their interpretation naturally depends on the chosen probe representations, such as bag-of-phones for $P$ and frozen LLM-based semantic embeddings for $S$. Future work will explore tighter non-linear estimators and validate these diagnostics across more encoder families.

\section{Training Details}
\label{append1}

\subsection{Training Data Statistics}
\label{appendA1}

Across the pretraining, alignment, and SFT stages, we use the same speech corpora, with only the training steps and objectives varying at each stage. Table~\ref{tab:training_data} summarizes detailed statistics of the training data, including language coverage, domain diversity, and overall scale. Our training corpora consist of annotated speech-text pairs totaling approximately 560K hours. 

% Most of the data are drawn from publicly available datasets, with a small fraction from in-house collections. This allows the research community to validate our training recipe using open-source resources.

\begin{table*}[ht]
\centering
\caption{An overview of the Mandarin and English speech corpora used across all training stages.}
\label{tab:training_data}
\resizebox{0.95\textwidth}{!}{
\begin{tabular}{l|c|c|c}
\toprule
\textbf{Dataset} & \textbf{Language(s)} & \textbf{Domain} & \textbf{Hours} \\
\midrule
YODAS-Granary~\citep{koluguri2025granary}   & English             & Variety        & 120K \\
Emilia~\citep{he2024emilia}               & English / Mandarin  & Variety   & 189K \\
MLS~\citep{pratap2020mls}                  & English             & Audiobook      & 45K \\
MSR-86K~\citep{li2024msr}              & English             & YouTube           &  10K \\
Common-Voice-v15~\citep{ardila2020common}     & English / Mandarin  & Read           & 3K \\
GigaSpeech~\citep{chen2021gigaspeech}           & English             & Variety        & 10K \\
LibriHeavy~\citep{kang2024libriheavy}           & English             & Audiobook      & 50K \\
LibriSpeech~\citep{panayotov2015librispeech}          & English             & Audiobook      & 960 \\
SPGISpeech~\citep{o2021spgispeech}           & English             & Finance        & 5K \\
PeopleSpeech~\citep{galvez2021people}         & English   & Variety        & 30K \\
VCTK~\citep{yamagishi2019cstr}                 & English             & Read           & 25 \\
TEDLIUM3~\citep{hernandez2018ted}          & English             & Talk           & 500 \\
WeNetSpeech-Yue~\citep{li2026wenetspeech}          & Chinese dialects            & Variety        & 22K \\
WeNetSpeech-Chuan~\citep{dai2025wenetspeech}          & Chinese dialects             & Variety        & 10K \\
WenetSpeech~\citep{zhang2022wenetspeech}          & Mandarin            & Variety        & 11K \\
FLEURS~\citep{conneau2023fleurs}               & English / Mandarin  & News           & 100 \\
AISHELL-1~\citep{bu2017aishell}            & Mandarin            & Read           & 150 \\
AISHELL-2~\citep{du2018aishell}            & Mandarin            & Read           & 1K \\
KeSpeech~\citep{tang2021kespeech}             & Mandarin and 8 Subdialects            & Conversation   & 1.6K \\
CS-Dialogue~\citep{zhou2025cs}          & Mandarin English Code Switch            & Variety        & 104 \\
ASCEND~\citep{lovenia2022ascend}          & Mandarin English Code Switch             & Conversation        & 10 \\
In-house data        &  English / Mandarin  &  Conversation  & $\sim\!$ 50K \\
\midrule
Total               & English / Mandarin  & All           & $\sim\!$ 560K \\
\bottomrule
\end{tabular}
}
\end{table*}

\subsection{Our LLM-based ASR Architecture}
\label{appendA2}

% \begin{itemize}
% \item \textbf{Feature Extraction.} 
\textbf{Feature extraction.}
We extract 80-dimensional log-Mel spectrograms using a 25ms window and a 10ms frame shift, followed by global mean and variance normalization.

% \item \textbf{Speech encoder.}
\textbf{Speech encoder.}
The backbone of our encoder is inherited from FireRedASR-AED~\citep{xu2025fireredasr}, consisting of a 4$\times$ downsampling convolutional module followed by a stack of Conformer blocks~\citep{gulati2020conformer}, with a total of approximately 600M parameters. The encoder converts speech into continuous representations at a frame rate of 25~Hz (40~ms temporal resolution). 

\textbf{CTC head.}
For encoder pretraining, we attach a three-layer MLP as a CTC head after the speech encoder, which maps its hidden representations to the target vocabulary and is optimized with the connectionist temporal classification loss~\citep{graves2006connectionist,yao2024cr}. The CTC head is used exclusively during pretraining.

%\item \textbf{Speech adaptor.}
\textbf{Speech adaptor.}
A lightweight speech adaptor, composed of an MLP with two linear layers, is responsible for mapping the encoder’s speech representations into the text embedding space of the LLM. Prior to the projection, a 4$\times$ downsampling is applied by concatenating 4 consecutive frames along the feature dimension to reduce the sequence length. Following this downsampling process, the frame rate is reduced to 6.25~Hz, corresponding to a temporal resolution of 160~ms per token.

%\item \textbf{LLM decoder.}
\textbf{LLM decoder.}
The decoder is initialized from Qwen3-1.7B~\citep{yang2025qwen3} and generates the final transcription conditioned on both text prompts (``Transcribe the speech into text.'') and speech tokens. 

%\end{itemize}

\subsection{Training Setups}
\label{appendA3}

In this study, we adopt stage-specific training configurations described as follows:

\textbf{Pretraining:} A dynamic batch strategy based on speech frames is employed, with a maximum of 20k frames per batch to efficiently handle variable-length utterances. The learning rate follows a schedule with a linear warm-up to 5.0e-4 over the first 8k steps, followed by exponential decay for the remainder of training.

\textbf{Alignment:} The batch size is set to 10k frames. The learning rate follows the same schedule as in pretraining, with a linear warm-up to a maximum of 1.0e-3.

\textbf{IA-SFT and joint SFT:} The batch size is further reduced to 7k frames. The learning rate also follows the same warm-up and exponential decay schedule, with a maximum learning rate of 1.0e-5. In pilot experiments, we also explored lower learning rates during joint SFT to reduce representation drift. Empirically, after IA-SFT has sufficiently narrowed the modality gap, a learning rate of 1.0e-5 maintained the interface stability established in prior stages and performed better than the more conservative alternatives.

The Adam optimizer~\citep{kingma2014adam} is used across all training stages. Experiments are conducted on NVIDIA A100 80GB GPUs using DeepSpeed ZeRO Stage-2~\citep{rajbhandari2020zero}, with gradient accumulation over 4 steps, bfloat16 precision, and FlashAttention-2~\citep{dao2023flashattention}.

\subsection{Training Details of IA-SFT}
\label{appendA4}

%We propose an IA-SFT strategy to maintain functional decoupling between the encoder and the LLM during post-training. 
Here, we provide additional details on the implementation of IA-SFT. %Figure~\ref{fig_sft} illustrates the IA-SFT workflow and compares it with standard SFT, while Table~\ref{tab:training_schedule} presents the training details of the IA-SFT, including the encoder updating schedule, the CKA threshold settings, and the specific training steps for each stage.

\paragraph{CKA-guided encoder update schedule.}
During IA-SFT, we perform several rounds of encoder hot-swapping to ensure diversity in the representations exposed to the LLM. Meanwhile, we aim to limit representation drift, as measured by CKA, so that the adaptor--LLM interface does not need to spend excessive effort adapting to each newly swapped encoder. To this end, we monitor changes in the representation distribution on a fixed validation set using CKA scores, as described in Section~\ref{sec3.4}. 

Given the current encoder checkpoint $\mathcal{E}_{\mathrm{cur}}$ and a reference checkpoint $\mathcal{E}_{\mathrm{ref}}$, we compute the CKA score between their representations and trigger an update once the score falls below a predefined threshold $\tau$. 
% \begin{equation}
% \mathrm{CKA}(\mathcal{E}_{\mathrm{cur}}, \mathcal{E}_{\mathrm{ref}}) < \tau ,
% \end{equation}
Given two sets of representation $E^{(a)}, E^{(b)} \in \mathbb{R}^{L \times d_e}$ extracted from the same input by two encoders $\mathcal{E}_{\mathrm{cur}}$ and $\mathcal{E}_{\mathrm{ref}}$, CKA is defined as
\begin{equation}
\text{CKA}(E^{(a)}, E^{(b)}) = \frac{ \langle \tilde{K}^{(a)}, \tilde{K}^{(b)} \rangle_F }{ \sqrt{ \langle \tilde{K}^{(a)}, \tilde{K}^{(a)} \rangle_F \cdot \langle \tilde{K}^{(b)}, \tilde{K}^{(b)} \rangle_F } },
\label{eq:cka}
\end{equation}
where $\tilde{K}^{(a)}$ and $\tilde{K}^{(b)}$ are centered Gram matrices computed as $\tilde{K}^{(x)} = \tilde{H}E^{(x)} E^{(x)\top}\tilde{H}$. The centering matrix is defined as $\tilde{H} = I_L - \frac{1}{L}J_L$, where $I_L$ is the identity matrix and $J_L$ is the all-ones matrix. CKA measures the geometric similarity of representation spaces, invariant to orthogonal transformations and isotropic scaling.

\paragraph{Iterative schedule and hot-swapping.} 
Whenever an update is triggered, the current pretraining encoder $\mathcal{E}_{\mathrm{cur}}$ is used to simultaneously update both (i) the frozen encoder $\mathcal{E}^{\mathrm{SFT}}$ in the IA-SFT pipeline, and (ii) the reference checkpoint $\mathcal{E}_{\mathrm{ref}}$ used for subsequent CKA monitoring during pretraining:
\begin{equation}
\mathcal{E}^{\mathrm{SFT}} \leftarrow \mathcal{E}_{\mathrm{cur}}, \qquad
\mathcal{E}_{\mathrm{ref}} \leftarrow \mathcal{E}_{\mathrm{cur}}.
\end{equation}
% and continue SFT without re-running alignment.
% This yields an iterative sequence
% \begin{equation}
% \mathcal{E}_{\text{init}} \rightarrow \mathcal{E}_1 \rightarrow \mathcal{E}_2 \rightarrow \cdots \rightarrow \mathcal{E}_{\text{final}},
% \end{equation}
% while the adaptor and LLM are continuously optimized.
We begin monitoring CKA scores every 10k steps once pretraining reaches 500k steps, at which point the model starts to exhibit convergence trends. The encoder at 500k steps serves as the initial reference checkpoint $\mathcal{E}_{\mathrm{ref}}$. Since pretraining focuses solely on the ASR task, the optimization direction remains largely consistent, causing CKA scores to generally decrease, with occasional rebounds throughout encoder evolution---this partly explains why direct encoder hot-swapping during IA-SFT works effectively without requiring realignment. Based on our experience, we set the CKA threshold $\tau=0.975$, which we find to be a moderate choice: a higher threshold triggers more frequent updates, incurring unnecessary overhead; a lower threshold leads to less frequent updates, permitting greater representation drift between swaps, which may reduce the LLM's ability to capture robust patterns shared across encoder states.

When pretraining reaches 1.01M steps, the CKA score first drops below this threshold (see Figure~\ref{fig_cka1}). At this point, the corresponding encoder is used to initialize the encoder in our encoder--adaptor--LLM model, and the alignment stage begins. After 1.3M alignment steps, we proceed to the IA-SFT stage, where the adaptor and LLM are jointly optimized. From then on, IA-SFT and pretraining are executed asynchronously in parallel. When pretraining reaches 1.32M steps, the CKA score again drops below 0.975 (see Figure~\ref{fig_cka2}). We then directly update both $\mathcal{E}^{\mathrm{SFT}}$ and $\mathcal{E}_{\mathrm{ref}}$, and continue asynchronous training. Finally, when pretraining reaches the maximum step of 2M (see Figure~\ref{fig_cka3}), we perform the last encoder hot-swapping for IA-SFT. As summarized in Table~\ref{tab:training_schedule}, the SFT-stage encoder is initialized once and updated twice throughout the process.

\begin{figure}[t]
    \centering
    \begin{subfigure}[t]{0.48\linewidth}
        \centering
        \includegraphics[width=\linewidth]{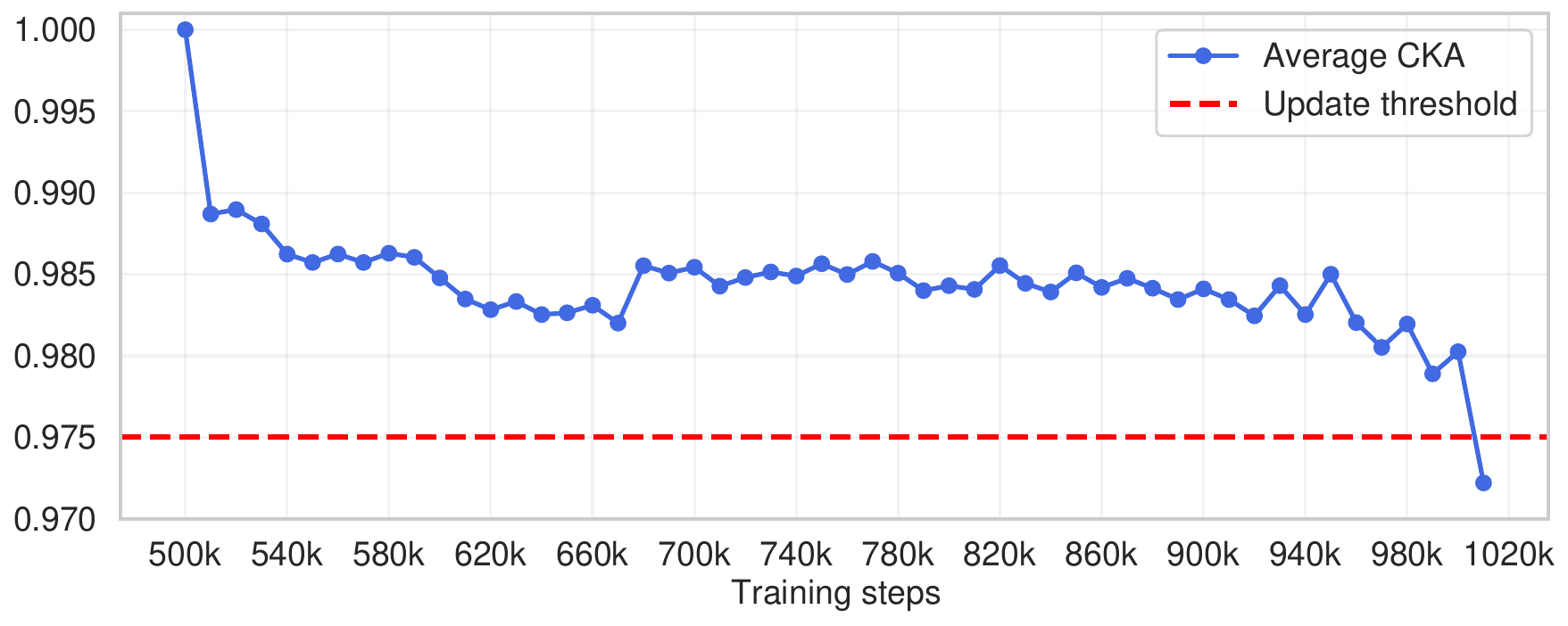}
        \caption{CKA scores between encoders (0.50~M-1.01~M) and reference checkpoint (0.50~M)}
        \label{fig_cka1}
    \end{subfigure}
    \hfill
    \begin{subfigure}[t]{0.48\linewidth}
        \centering
        \includegraphics[width=\linewidth]{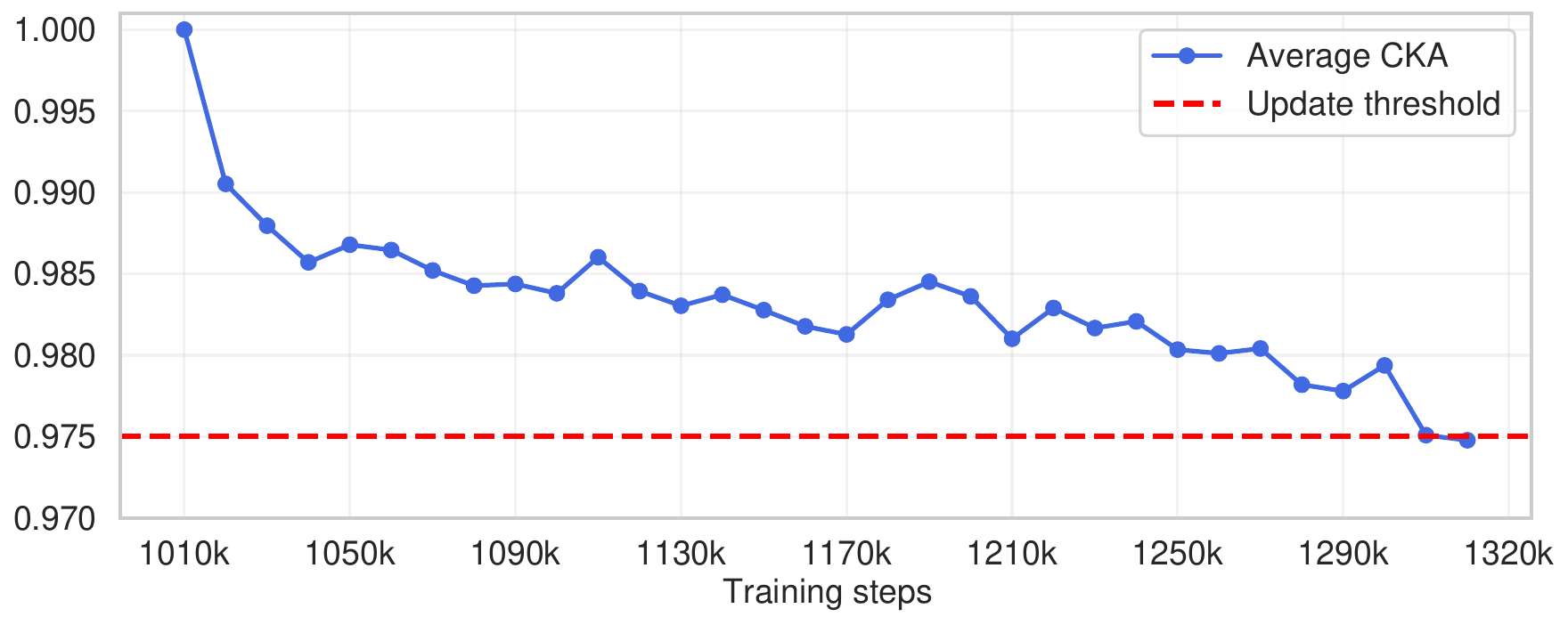}
        \caption{CKA scores between encoders (1.01~M-1.32~M) and reference checkpoint (1.01~M)}
        \label{fig_cka2}
    \end{subfigure}
    \hfill
    \begin{subfigure}[b]{0.98\linewidth}
        \centering
        \includegraphics[width=\linewidth]{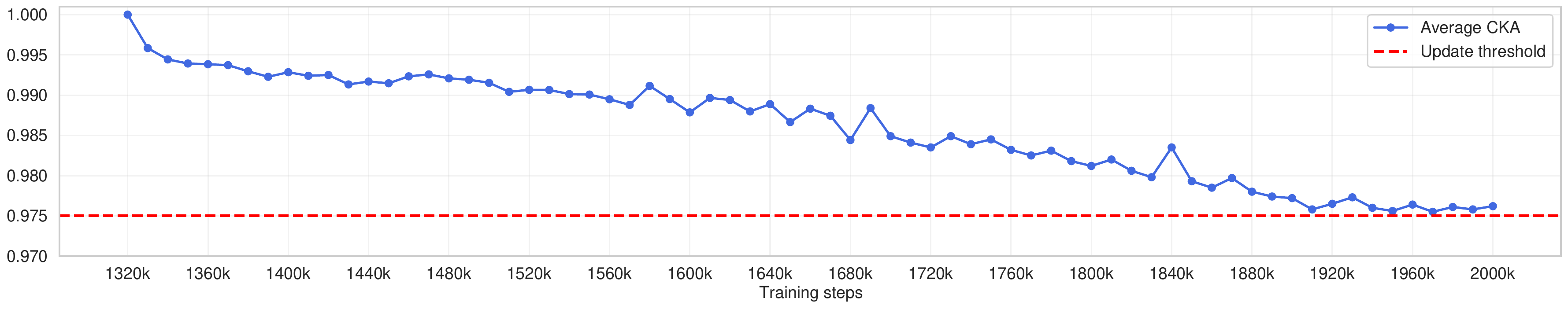}
        \caption{CKA scores between encoders (1.32~M-2.00~M) and reference checkpoint (1.32~M)}
        \label{fig_cka3}
    \end{subfigure}
    \caption{Trajectory of CKA scores during pretraining. It reports the average CKA between the encoders and the corresponding reference checkpoint.}
    \label{fig_cka}
    \vskip -2mm
\end{figure}

\begin{table*}[!h]
    \centering
    \caption{Detailed training procedures, including the encoder update schedule, CKA thresholds, and stage-wise training steps. Here, the step counts for pretraining, alignment, and IA-SFT are measured independently in their respective training processes.}
    \label{tab:training_schedule}
    \renewcommand{\arraystretch}{1.25}
    \setlength{\tabcolsep}{8pt}
    \begin{tabularx}{\linewidth}{@{}>{\centering\arraybackslash}p{0.13\linewidth}
                                   >{\centering\arraybackslash}p{0.16\linewidth}
                                   >{\raggedright\arraybackslash}X@{}}
        \toprule
        \textbf{Pretraining Step} & \textbf{Trigger} & \textbf{Action Details} \\
        \midrule
        $0.50$\,M  & -- &
        \textbf{\textcircled{\small 1} Pretraining begins.} The reference checkpoint is initialized at the $0.50$\,M pretraining step. \\
        \addlinespace[2pt]
        \midrule
        \addlinespace[2pt]
        $1.01$\,M  & $\mathrm{CKA} < 0.975$ &
        \textbf{\textcircled{\small 2} Pretraining continues.} The encoder snapshot updates the reference checkpoint and initializes the encoder for post-training. The alignment stage then begins and runs for $1.3$\,M steps in total. Upon its completion, IA-SFT begins and runs for $1.0$\,M steps, while pretraining continues asynchronously. \\
        \addlinespace[2pt]
        \midrule
        \addlinespace[2pt]
        $1.32$\,M  & $\mathrm{CKA} < 0.975$ &
        \textbf{\textcircled{\small 3} Pretraining continues.} The encoder snapshot updates both the reference checkpoint and the encoder for IA-SFT, which then runs for another $1.0$\,M steps. \\
        \addlinespace[2pt]
        \midrule
        \addlinespace[2pt]
        $2.00$\,M (end) & -- &
        \textbf{\textcircled{\small 4} Pretraining ends.} The encoder snapshot updates the encoder for IA-SFT, which then runs for another $2.0$\,M steps. \\
        \addlinespace[2pt]
        \midrule
        \addlinespace[2pt]
        --  & -- &
        \textbf{\textcircled{\small 5} IA-SFT ends}, followed by joint-SFT, which runs for $2.00$\,M steps in total. \\
        \bottomrule
    \end{tabularx}
\end{table*}

\paragraph{Why Encoder Updates without Realignment Work?}
First, in our design, pretraining and SFT target the same supervised ASR task, differing only in the supervision signals and loss functions, which helps maintain consistency in the overall optimization direction. Moreover, as shown in Figure~\ref{fig_cka}, CKA scores between encoders remain consistently high throughout pretraining, suggesting substantial similarity in the dominant encoder subspace across checkpoints. In this regime, subsequent pretraining appears to mainly reduce acoustic redundancy and refine acoustic preferences, rather than reshape the global manifold. Whenever the updated encoder is applied to the SFT model, downstream components receive representations with lower entropy while remaining broadly distributionally consistent. The adaptor and LLM can therefore continue adapting to these inputs with limited additional alignment cost, which may help explain why explicit realignment was not necessary in our experiments.

In general, compared to the traditional pipeline, IA-SFT provides the following main advantages:

\begin{itemize}[leftmargin=1.6em, itemsep=2pt, parsep=0pt, topsep=1pt]
    \item \textbf{Maintaining functional decoupling.}
    In contrast to sequential SFT with trainable encoders, IA-SFT helps limit the tendency of the encoder to drift toward the LLM's semantic manifold. The encoder can therefore continue optimizing primarily for acoustic uncertainty reduction, which is more consistent with preserving fine-grained acoustic discrimination and limiting semantic dependence.

    \item \textbf{Regularization from multiple encoder perspectives.}
    Beyond the reduction in entropy, the iteratively updated encoders exhibit distinct preference patterns for acoustic cues. IA-SFT exposes the LLM to representations with different signal-to-noise levels in a curriculum-like manner. This process can also act as a form of regularization, encouraging the LLM to focus on acoustic features that remain consistent across encoder states rather than overfitting to idiosyncratic biases of any single checkpoint.

    \item \textbf{Parallelization and efficiency.}
    Converting sequential pretraining and SFT into asynchronous parallel processes substantially reduces the total training time.

    \item \textbf{Flexibility and transferability.} Benefiting from functional decoupling, our encoder remains compatible with both the CTC head and the LLM decoder in our setup. This may improve deployment flexibility by enabling scalable configurations tailored to different computational constraints based on a shared encoder. It also creates a cleaner interface for iterating on individual components, which may reduce the cost of integrating newer LLM backbones in future system updates.

\end{itemize}

\section{Formalization of Encoder Representation Dynamics}
\label{sec:appendix_dynamics}

This section provides a geometric perspective on encoder representation dynamics in LLM-based ASR systems with an encoder--adaptor--LLM architecture. 
Our goal is to characterize how different pretraining and post-training paradigms shape the spectral structure and functional composition of encoder representations, thereby inducing distinct operating regimes at the encoder--LLM interface.

\subsection{Problem Setup and Spectral Decomposition}

Let $X$ be a random variable representing input speech features, and let $\mathcal{E}_\phi$ denote a speech encoder. 
For a realization $x$, the encoder produces
\begin{equation}
E = \mathcal{E}_\phi(x) = [e_1,\ldots,e_L]^\top \in \mathbb{R}^{L \times d_e},
\end{equation}
where $L$ is the sequence length and $d_e$ is the hidden dimension. Let $E'\in \mathbb{R}^{L_{valid} \times d_e}$ denote the valid-frame representation obtained from $E$ after removing padded positions according to the sequence mask.
Consider the singular value decomposition:
\begin{equation}
E' = U \Sigma V^\top,
\end{equation}
where $\Sigma=\mathrm{diag}(\sigma_1,\ldots,\sigma_d)$ with $d=\min(L_{valid},d_e)$.

The singular values $\{\sigma_i\}$ characterize how variance is distributed across principal directions. 
By normalizing the singular values and treating them as a discrete spectrum, we can further quantify their concentration using an entropy measure, leading to the normalized spectral entropy (NSE).
Moreover, a concentrated spectrum corresponds to a lower-entropy and more compact representation, whereas a flatter spectrum indicates higher residual uncertainty distributed across multiple directions. 
Since $E$ is the encoder-side representation feeding into the adaptor--LLM stack, its spectral structure reflects how much uncertainty has already been reduced by the encoder before reaching the interface.

\subsection{Functional Subspaces of Speech Representations}

We view the encoder feature space as composed of multiple overlapping functional subspaces:
\begin{equation}
\mathbb{R}^{d_e}
\;\supset\;
\mathcal{S}_{\text{linguistic}}
\;\oplus\;
\mathcal{S}_{\text{paralinguistic}}
\;\oplus\;
\mathcal{S}_{\text{non-linguistic}},
\end{equation}
where
\begin{itemize}[leftmargin=1.6em, itemsep=2pt, parsep=0pt, topsep=1pt]
    \item $\mathcal{S}_{\text{linguistic}}$ captures transcription-relevant structure (e.g., phonetic and lexical information);
    \item $\mathcal{S}_{\text{paralinguistic}}$ includes speaker prosody and emotion;
    \item $\mathcal{S}_{\text{non-linguistic}}$ corresponds to environmental noise and nuisance variability.
\end{itemize}

These subspaces are not strictly orthogonal, but this decomposition provides a useful abstraction for analyzing how training redistributes variance. 
An effective ASR encoder should concentrate variance within $\mathcal{S}_{\text{linguistic}}$ while suppressing irrelevant variability, thereby forming a compact and acoustically grounded interface.

\subsection{Accessible-Information Proxies under a Gaussian Approximation}

Beyond spectral concentration, we also wish to characterize how much transcription-relevant information remains accessible from the encoder representation. 
Let $u$ denote an utterance-level summary of the encoder representation, $P$ a phonetic target variable, and $S$ a semantic target variable. 
To quantify their statistical dependence in a tractable form, we adopt a joint Gaussian approximation on $(u,P,S)$ and use mutual-information-inspired log-determinant quantities as accessible-information proxies.

For a Gaussian random vector $\mathbf{r} \in \mathbb{R}^k$ with covariance $\Sigma_\mathbf{r}$, its differential entropy is given by
\begin{equation}
h(\mathbf{r})=\frac{1}{2}\log\!\big((2\pi e)^k \det \Sigma_\mathbf{r}\big).
\end{equation}
This shows that, under a Gaussian assumption, entropy is fully determined by the covariance structure through its log-determinant.

\begin{proposition}[Gaussian mutual information in log-det form]
Let $A \in \mathbb{R}^{d_A}$ and $B \in \mathbb{R}^{d_B}$ be jointly Gaussian random variables with joint covariance
\begin{equation}
\Sigma_{[A,B]}=
\begin{bmatrix}
\Sigma_{AA} & \Sigma_{AB} \\
\Sigma_{BA} & \Sigma_{BB}
\end{bmatrix}.
\end{equation}
Then their mutual information admits the closed-form expression
\begin{equation}
I(A;B)
=
\frac{1}{2}\log
\frac{\det \Sigma_{AA}\,\det \Sigma_{BB}}
     {\det \Sigma_{[A,B]}}.
\end{equation}
\end{proposition}

\begin{proof}
By definition,
\begin{equation}
I(A;B)=h(A)+h(B)-h(A,B).
\end{equation}
Since $(A,B)$ is jointly Gaussian, its marginals are also Gaussian. Therefore,
\begin{align}
h(A)
&=
\frac{1}{2}\log\!\big((2\pi e)^{d_A}\det \Sigma_{AA}\big),\\
h(B)
&=
\frac{1}{2}\log\!\big((2\pi e)^{d_B}\det \Sigma_{BB}\big),\\
h(A,B)
&=
\frac{1}{2}\log\!\big((2\pi e)^{d_A+d_B}\det \Sigma_{[A,B]}\big).
\end{align}
Substituting these expressions into the mutual information identity gives
\begin{align}
I(A;B)
&=
\frac{1}{2}\log\!\big((2\pi e)^{d_A}\det \Sigma_{AA}\big)
+
\frac{1}{2}\log\!\big((2\pi e)^{d_B}\det \Sigma_{BB}\big)\nonumber\\
&\quad-
\frac{1}{2}\log\!\big((2\pi e)^{d_A+d_B}\det \Sigma_{[A,B]}\big).
\end{align}
Using logarithm rules, the constant terms cancel, yielding
\begin{equation}
I(A;B)
=
\frac{1}{2}\log
\frac{\det \Sigma_{AA}\,\det \Sigma_{BB}}
     {\det \Sigma_{[A,B]}}.
\end{equation}
\end{proof}

Applying this result to $(u,P)$ yields a Gaussian proxy for phonetic accessible information:
\begin{equation}
\mathrm{PAI}(E')
\;\propto\;
I(u;P)
=
\frac{1}{2}\log
\frac{\det \Sigma_{uu}\,\det \Sigma_{PP}}
     {\det \Sigma_{[u,P]}}.
\end{equation}

\begin{proposition}[Gaussian conditional mutual information in log-det form]
Let $A \in \mathbb{R}^{d_A}$, $B \in \mathbb{R}^{d_B}$, and $C \in \mathbb{R}^{d_C}$ be jointly Gaussian. Then
\begin{equation}
I(A;B\mid C)
=
\frac{1}{2}\log
\frac{\det \Sigma_{AA\mid C}\,\det \Sigma_{BB\mid C}}
     {\det \Sigma_{[A,B]\mid C}},
\end{equation}
where the conditional covariance matrices are defined via the Schur complement, e.g.
\begin{equation}
\Sigma_{AA\mid C}
=
\Sigma_{AA}-\Sigma_{AC}\Sigma_{CC}^{-1}\Sigma_{CA},
\end{equation}
and analogously for $\Sigma_{BB\mid C}$ and $\Sigma_{[A,B]\mid C}$. 
Here, $\Sigma_{AA\mid C}$ captures the residual variability of $A$ after removing the components that can be linearly explained by $C$.
\end{proposition}

\begin{proof}
By definition,
\begin{equation}
I(A;B\mid C)=h(A\mid C)+h(B\mid C)-h(A,B\mid C).
\end{equation}
For jointly Gaussian variables, conditional distributions remain Gaussian, and their conditional entropies are determined by conditional covariance matrices:
\begin{align}
h(A\mid C)
&=
\frac{1}{2}\log\!\big((2\pi e)^{d_A}\det \Sigma_{AA\mid C}\big),\\
h(B\mid C)
&=
\frac{1}{2}\log\!\big((2\pi e)^{d_B}\det \Sigma_{BB\mid C}\big),\\
h(A,B\mid C)
&=
\frac{1}{2}\log\!\big((2\pi e)^{d_A+d_B}\det \Sigma_{[A,B]\mid C}\big).
\end{align}
Substituting into the definition of conditional mutual information yields
\begin{align}
I(A;B\mid C)
&=
\frac{1}{2}\log\!\big((2\pi e)^{d_A}\det \Sigma_{AA\mid C}\big)
+
\frac{1}{2}\log\!\big((2\pi e)^{d_B}\det \Sigma_{BB\mid C}\big)\nonumber\\
&\quad-
\frac{1}{2}\log\!\big((2\pi e)^{d_A+d_B}\det \Sigma_{[A,B]\mid C}\big).
\end{align}
Again, the constant terms cancel, giving
\begin{equation}
I(A;B\mid C)
=
\frac{1}{2}\log
\frac{\det \Sigma_{AA\mid C}\,\det \Sigma_{BB\mid C}}
     {\det \Sigma_{[A,B]\mid C}}.
\end{equation}
\end{proof}

By setting $(A,B,C)=(u,S,P)$, we obtain the conditional semantic accessible information:
\begin{equation}
\mathrm{CSAI}(E')
\;\propto\;
I(u;S\mid P)
=
\frac{1}{2}\log
\frac{\det \Sigma_{uu\mid P}\,\det \Sigma_{SS\mid P}}
     {\det \Sigma_{[u,S]\mid P}}.
\end{equation}

% In practice, we estimate covariance matrices empirically with ridge regularization and convert the logarithm to base 2, leading to the regularized forms in Eqs.~\eqref{eq:pai} and~\eqref{eq:csai}. 
% These quantities should be interpreted as linear-Gaussian accessible-information proxies for relative comparison, rather than exact mutual information estimates.

In practice, we estimate covariance matrices empirically with ridge regularization and convert the logarithm to base 2,
leading to the regularized forms in Eqs.~\eqref{eq:pai} and~\eqref{eq:csai}. 
These quantities should be interpreted as linear-Gaussian accessible-information proxies for relative comparison, rather than exact mutual information estimates. Unlike PAI and CSAI, NSE measures the entropy of the normalized singular-value spectrum, reflecting the spectral concentration and geometric compactness of encoder representations. Together, these metrics provide complementary perspectives on the evolution of encoder representations during training.

\subsection{Encoder Pretraining and Spectral Bias}

Different pretraining objectives induce distinct spectral biases, shaping how acoustic uncertainty is distributed across the encoder--LLM interface.

\paragraph{Character-level CTC.}
Character-level CTC enforces monotonic alignment between speech frames and transcription labels:
\begin{equation}
\mathcal{L}_{\text{CTC-char}}
=
-\log \sum_{\pi \in \mathcal{B}^{-1}(y^{\text{char}})}
\prod_{i=1}^{L} p_{\theta}(\pi_i \mid e_i).
\end{equation}
where $y^{\text{char}}$ denotes the target character sequence, $\pi$ is a frame-level alignment path, and $\mathcal{B}$ is the collapse operator.
This objective encourages compact, transcription-aligned representations, but introduces early coupling to language-specific semantic units.

\paragraph{Phoneme-level CTC.}
Phoneme-level supervision mitigates premature semantic coupling:
\begin{equation}
\mathcal{L}_{\text{CTC-phoneme}}
=
-\log \sum_{\pi \in \mathcal{B}^{-1}(y^{\text{phoneme}})}
\prod_{i=1}^{L} p_{\theta}(\pi_i \mid e_i).
\end{equation}
where $y^{\text{phoneme}}$ denotes the phoneme sequence.
As phonemes are more directly tied to acoustic structure, this objective yields more language-agnostic and acoustically grounded representations, forming a more stable encoder--LLM interface.

\paragraph{AED-based pretraining.} The AED objective formulates speech modeling as sequence-to-sequence prediction: \begin{equation} \mathcal{L}_{\text{AED}} = -\sum_{n=1}^{N} \log p_{\theta}(y_n \mid y_{<n}, E) \end{equation} where the decoder attends to the encoder representation sequence. Compared with CTC, AED does not impose an explicit frame-level alignment constraint. Instead, it allows the encoder to retain broader contextual cues that can assist the autoregressive decoder. Moreover, the sequence-to-sequence formulation can be naturally extended beyond pure transcription to other speech-conditioned tasks. Consequently, although AED-pretrained encoders can support ASR, their representations tend to preserve a broader range of variability and exhibit less concentrated spectra than those optimized under strongly alignment-driven objectives. This broader representational capacity makes AED-pretrained encoders well suited for initializing LALMs, where the encoder is expected to support diverse audio--language tasks beyond transcription.

\paragraph{Hybrid supervised objectives.}
Some approaches combine alignment-driven and sequence-level objectives (e.g., CTC + AED):
\begin{equation}
\mathcal{L}_{\text{hybrid}} = \lambda \mathcal{L}_{\text{CTC}} + (1 - \lambda)\mathcal{L}_{\text{AED}}.
\end{equation}
Such hybrid training balances strong alignment constraints with sequence-level modeling, yielding intermediate spectral characteristics and uncertainty allocation.

\paragraph{Self-supervised pretraining (SSL).}
Self-supervised objectives such as Best-RQ~\citep{chiu2022self} learn acoustic representations by predicting discrete pseudo-targets derived from the input signal:
\begin{equation}
\mathcal{L}_{\text{SSL}}
=
-\sum_{t \in \mathcal{M}} \log p_{\theta}(z_t \mid \tilde{E}),
\end{equation}
where $z_t$ denotes discrete targets obtained via a quantization process, $\tilde{E}$ represents masked or corrupted acoustic features, $\mathcal{M}$ denotes the set of masked positions, and the model predicts a categorical distribution over codebook entries at masked positions.
In this framework, continuous speech is first mapped to discrete codebook indices, which serve as prediction targets, and the model is trained to infer these targets at masked positions from surrounding context. 
Unlike supervised objectives, SSL does not enforce alignment to linguistic units, resulting in higher-entropy representations and deferring more uncertainty to downstream modules such as the LLM.

\subsection{Instruction-Based Post-training and Divergent Encoder Regimes}

After encoder pretraining, both LLM-ASR and LALM systems are typically further optimized under instruction-conditioned language modeling objectives. Let $E=\{e_i\}_{i=1}^{L}$ denote the encoder representations, $t$ the text instruction prompt, and $y=(y_1,\ldots,y_N)$ the target response sequence. A generic post-training objective can be written as
\begin{equation}
\mathcal{L}_{\text{inst}}
=
-\sum_{n=1}^{N}
\log p_{\theta}(y_n \mid y_{<n}, E, t).
\end{equation}
Once the encoder is connected to the adaptor--LLM stack under such supervision, the downstream language modeling objective further reshapes the representation geometry at the encoder--LLM interface. 
In LLM-ASR, the encoder is usually initialized from an ASR-oriented model trained with transcription objectives such as CTC or AED, so its representations are already biased toward transcription-relevant acoustic structure and may provide a relatively compact interface before LLM integration. 
Moreover, the subsequent instruction-based supervision is still centered on transcription, which further encourages the encoder--LLM interface to remain specialized for ASR. 
Importantly, when a well-formed interface has been established prior to full end-to-end coupling, subsequent joint optimization primarily induces local refinements on the existing representation manifold, rather than causing large distribution shifts toward text-dominated regimes. 
This allows the model to improve cross-module alignment while preserving the capability boundary between acoustic grounding and semantic modeling.

By contrast, LALMs are typically trained to support a broader range of audio--language tasks beyond transcription. 
Their encoders therefore tend to preserve a wider range of information and exhibit a flatter, higher-entropy representation space. 
Moreover, the diverse instructions further encourage representations to remain compatible with more general audio understanding. 
Consequently, although both LLM-ASR and LALMs may adopt the same encoder--adaptor--LLM architecture, differences in both encoder initialization and post-training supervision lead them to evolve under distinct representational regimes.

\section{Entropy Allocation and Functional Decoupling in LLM-ASR}
\label{sec:appendix_IT}

In this section, we provide an information-theoretic perspective on entropy allocation in LLM-based ASR. 
Our goal is to clarify how a well-formed encoder--LLM interface enables a capability-aligned division of uncertainty reduction between acoustic grounding and semantic disambiguation.

% \subsection{Entropy Decomposition at the Encoder--LLM Interface}

% Let $X$ and $Y$ denote the input speech and target transcription, respectively. 
% The speech encoder and adaptor define a deterministic transformation chain
% \begin{equation}
% X \xrightarrow{\;\mathcal{E}_\phi\;} E \xrightarrow{\;\mathcal{A}_\psi\;} Z,
% \end{equation}
% where $E$ is the encoder representation and $Z$ is the projected speech embedding consumed by the LLM. 
% At the interface level, the uncertainty of the target transcription can be decomposed as
% \begin{equation}
% H(Y) = I(Y;Z) + H(Y \mid Z),
% \end{equation}
% where $I(Y;Z)$ measures the amount of task-relevant information exposed through the interface, and $H(Y \mid Z)$ represents the residual uncertainty to be resolved by the LLM.

% This decomposition naturally reflects a division of labor across modules. 
% A more informative and structured interface increases $I(Y;Z)$ and reduces the burden on the LLM, whereas a weaker interface shifts more uncertainty to downstream language modeling. 
% Importantly, the effectiveness of the interface depends not only on the quantity of retained information, but also on whether that information is aligned with the functional roles of each module.

\subsection{Entropy Decomposition at the Encoder--LLM Interface}

Let $X$ and $Y$ denote the input speech and target transcription, respectively.
The speech encoder and adaptor define a deterministic transformation chain
\begin{equation}
X \xrightarrow{\;\mathcal{E}_\phi\;} E
\xrightarrow{\;\mathcal{A}_\psi\;} Z,
\end{equation}
where $E$ is the encoder representation and $Z$ is the projected speech embedding consumed by the LLM.
At the interface level, the uncertainty of the target transcription can be decomposed as
\begin{equation}
H(Y) = I(Y;Z) + H(Y \mid Z),
\end{equation}
where $I(Y;Z)$ measures the amount of task-relevant information exposed through the interface, and $H(Y \mid Z)$ represents the residual uncertainty to be resolved by the LLM.

Following the information-theoretic perspective of the Information Bottleneck (IB) principle and its extensions to deep neural networks~\citep{tishby1999information,tishby2015deep,shwartz2017opening},
the data processing inequality implies that
\begin{equation}
I(Y;Z) \leq I(Y;E) \leq I(Y;X).
\end{equation}
Accordingly, the task-relevant information loss across the encoder--adaptor boundary can be expressed as
\begin{equation}
\Delta_{\mathrm{adapt}}^Y
= I(Y;E) - I(Y;Z)
= I(Y;E \mid Z) \geq 0.
\end{equation}
A boundary is task-lossless when $\Delta_{\mathrm{adapt}}^Y=0$,
whereas $\Delta_{\mathrm{adapt}}^Y>0$ indicates the loss of task-relevant information.

This decomposition naturally reflects a division of labor across modules.
A more informative and structured interface increases $I(Y;Z)$ and reduces the burden on the LLM, whereas a weaker interface shifts more uncertainty to downstream language modeling.
However, preserving task-relevant information alone does not necessarily ensure an appropriate allocation of modeling responsibilities across heterogeneous modules.
Importantly, the effectiveness of the interface depends not only on the quantity of retained information, but also on whether that information is aligned with the functional roles of each module.

\subsection{Capability Boundary and Functional Decoupling}

In LLM-based ASR, the encoder and the LLM exhibit complementary inductive biases. 
The encoder is well suited for resolving local acoustic ambiguity, such as phonetic distinctions and temporal structure, while the LLM is more effective at resolving higher-level ambiguity through linguistic priors and contextual reasoning. 
Accordingly, a well-formed encoder--LLM interface should satisfy two properties. 
First, it should provide a sufficiently compact representation to avoid unnecessary burden on the LLM. 
Second, the retained structure should remain acoustically grounded, preserving evidence derived from the speech signal rather than replacing it with text-correlated shortcuts.
Under this perspective, the encoder and LLM operate under functional decoupling: the encoder primarily reduces acoustic uncertainty, while the LLM resolves the remaining semantic ambiguity conditioned on the interface representation.

\subsection{Hallucination as Misallocated Uncertainty Reduction}

From this viewpoint, hallucination in LLM-based ASR can be interpreted as a consequence of misallocated uncertainty reduction across the encoder--LLM interface. 
Two representative failure modes are particularly relevant.

\paragraph{Semantic-contaminated encoder representations.}
One failure mode arises when joint optimization progressively aligns encoder representations with text-correlated regularities. 
In this regime, part of the uncertainty reduction is achieved through patterns that are not strictly grounded in acoustic evidence. 
As a result, the interface may become more predictive of the transcription while losing robustness to acoustic variation, increasing the likelihood of fluent but weakly grounded outputs under ambiguous or degraded conditions.

\paragraph{LLM-dominant uncertainty reduction.}
Another failure mode occurs when the encoder provides a weak or insufficiently structured interface. 
In this case, a larger portion of uncertainty is deferred to the LLM, making the decoding process more dependent on language priors. 
This behavior can be further understood through the Bayesian factorization
\begin{equation}
p(Y\mid Z) \propto p(Z \mid Y)\, p(Y).
\end{equation}
When $Z$ provides limited discrimination among candidate transcriptions, the likelihood term $p(Z\mid Y)$ becomes less informative, and the posterior is increasingly shaped by the prior $p(Y)$. 
Consequently, the model may generate outputs that are linguistically plausible but not fully supported by the speech signal.

\subsection{Rationale of the Decoupled Training Paradigm}

The above analysis motivates a training strategy that aligns uncertainty allocation with module capabilities. 
Phoneme-level pretraining encourages the encoder to form a compact and acoustically grounded interface by emphasizing pronunciation-level structure. 
Subsequent adaptation stages allow the adaptor and LLM to operate on this interface before full end-to-end coupling. 
After the interface has been sufficiently established, joint optimization can be introduced to refine the overall system. 
At this stage, interactions between modules primarily induce local adjustments in the representation manifold, improving alignment at the encoder--LLM interface without substantially altering the underlying representation structure.

From this perspective, hallucination can be understood as a consequence of how uncertainty reduction is distributed across modules. 
Maintaining an interface that is both compact and acoustically grounded preserves the intended division of labor between acoustic grounding and downstream language modeling. 
As a result, the encoder can absorb a larger portion of acoustically grounded uncertainty reduction, while leaving higher-level semantic disambiguation to the LLM. 
This capability-aligned allocation not only mitigates hallucination by reducing reliance on text-side priors, but also improves parameter efficiency by alleviating the burden on the LLM, enabling strong performance with a smaller model scale.

\end{document}